\documentclass[10pt]{article}

\usepackage[T1]{fontenc}										% font encoding ... definitely need this
\usepackage{lmodern}												% loads latin modern fonts ... not sure if needed but doesn't hurt to include it	
\usepackage{microtype}											% improves spacing
\usepackage[LGRgreek,frenchmath]{mathastext}% sets default math font to the default roman text font ... option [frenchmath] for lower case Roman fonts in italics
\usepackage[mathscr]{eucal}                 % gives you \mathscr{} font
\usepackage[sans]{dsfont}										% gives you \mathds{} font ... option [sans] for sans serif
\usepackage{bbold}													% gives you sans serif \mathbb{} font  ... may need to compile with dvi or ps first ... not sure
\usepackage{bbm}														% gives you various \mathbbm{} and \mathbbmss{} and \mathbbmtt{} fonts
\usepackage{graphicx}                      	% need for figures
\usepackage{subcaption}
\usepackage{mdwlist}												% need for itemize* (compact itemize)

\usepackage[authoryear]{natbib}  											% need for bibliography
\usepackage[dvipsnames]{xcolor}							% gives you color definitions
\usepackage[plainpages=false, pdfpagelabels]{hyperref} % enables and customizes hyperlinks
	\hypersetup{
		colorlinks   = true,
		citecolor    = RoyalBlue, %NavyBlue,
		linkcolor    = RubineRed, %Maroon,
		urlcolor     = Turquoise
	}
\usepackage[paperwidth=8.5in,paperheight=11.00in,top=1.00in, bottom=1.00in, left=0.75in, right=0.75in]{geometry}
\usepackage{mathtools}                      % need for `show only references'
\mathtoolsset{showonlyrefs=true}            % only equations which are labeled AND referenced will numbered...use \eqref{} instead of (\ref{})
\usepackage{amsmath}
\usepackage{amssymb}
\usepackage{amsfonts}
\usepackage{amsthm}                         % need for theorem-proof environment
\allowdisplaybreaks                         % allows page breaks for long equations...you can prevent a page-break with \\*
\newtheoremstyle{plain}
  {}   				% ABOVESPACE
  {}   				% BELOWSPACE
  {\itshape}  % BODYFONT
  {}       		% INDENT (empty value is the same as 0pt)
  {\mdseries\scshape} % HEADFONT
  {.}         % HEADPUNCT
  { } 				% HEADSPACE
  {\thmname{#1}\thmnumber{ #2}\ifx#3\empty\else\ (#3)\fi}
\theoremstyle{plain}
\newtheorem{theorem}{\underline{Theorem}}
\newtheorem{lemma}[theorem]{\underline{Lemma}}       	
\newtheorem{proposition}[theorem]{\underline{Proposition}}

\newtheoremstyle{definition}
  {}   				% ABOVESPACE
  {}   				% BELOWSPACE
  {}  				% BODYFONT
  {}      		% INDENT (empty value is the same as 0pt)
  {\mdseries\scshape} % HEADFONT
  {.}         % HEADPUNCT
  { } 				% HEADSPACE
  {\thmname{#1}\thmnumber{ #2}\ifx#3\empty\else\ (#3)\fi}
\theoremstyle{definition}
\newtheorem{definition}[theorem]{\underline{Definition}}
\newtheorem{example}[theorem]{\underline{Example}}
\newtheorem{remark}[theorem]{\underline{Remark}}
\newtheorem{assumption}[theorem]{\underline{Assumption}}

\usepackage{sectsty}
\allsectionsfont{\mdseries\scshape}

\renewcommand{\[}{\left[}

\newcommand\Eb{\mathds{E}}

\newcommand\Pb{\mathds{P}}

\newcommand\Rb{\mathds{R}}

\newcommand\Nb{\mathds{N}}

\newcommand\Ac{\mathscr{A}}

\newcommand\Ec{\mathscr{E}}

\newcommand\Hc{\mathscr{H}}

\newcommand\Lc{\mathscr{L}}

\newcommand\Nc{\mathscr{N}}

\newcommand\eps{\varepsilon}

\newcommand\sig{\sigma}

\newcommand\gam{\gamma}
\newcommand\Gam{\Gamma}
\newcommand\lam{\lambda}
\newcommand\del{\delta}

\renewcommand\phi{\varphi}

\newcommand\Hh{\widehat{H}}

\newcommand\Ebt{\widetilde{\Eb}}
\newcommand\Pbt{\widetilde{\Pb}}

\newcommand\Wt{\widetilde{W}}

\newcommand\pt{\widetilde{p}}

\renewcommand\d{\partial}

\newcommand\dd{\mathrm{d}}
\newcommand\ee{\mathrm{e}}

\begin{document}

\title{Optimal positioning in derivative securities in incomplete markets}

\author{
Tim Leung
\thanks{Department of Applied Mathematics, University of Washington.  \textbf{e-mail}: \url{timleung@uw.edu}}
\and
Matthew Lorig
\thanks{Department of Applied Mathematics, University of Washington.  \textbf{e-mail}: \url{mlorig@uw.edu}}
\and
Yoshihiro Shirai
\thanks{Department of Applied Mathematics, University of Washington.  \textbf{e-mail}: \url{yshirai@uw.edu}}
}

\date{This version: \today}

\maketitle

\begin{abstract}
\noindent
This paper analyzes a problem of optimal static hedging using derivatives in incomplete markets. The investor is assumed to have a risk exposure to two underlying assets. The hedging instruments are vanilla options written on a single underlying asset. The hedging problem is formulated as a utility maximization problem whereby the form of the optimal static hedge is determined. Among our results, a semi-analytical solution for the optimizer is found through variational methods for exponential, power/logarithmic, and quadratic utility. When vanilla options are available for each underlying asset, the optimal solution is related to the fixed points of a Lipschitz map. In the case of exponential utility, there is only one such fixed point, and subsequent iterations of the map converge to it.
\end{abstract}

%-----------------------------------------------------------------------------------
%
%       SECTION: 		Introduction
%
%-----------------------------------------------------------------------------------

\section{Introduction}
Static hedging portfolios are often useful for fully or partially replicating a contingent claim using vanilla derivatives, and for establishing no-arbitrage relationships or bounds for exotic options (see \cite{BreedenLitzenberger}, \cite{CarrEllisGupta}, \cite{HobsonLaurenceWang}, \cite{CarrNadtochiy}, \cite{CarrWu}, among others). Using traded derivatives is particularly crucial when the underlying asset has limited liquidity or is not traded. Once a static hedging portfolio is formed, no continuous monitoring of the underlying or dynamic rebalancing is needed. Therefore, static hedging is a robust alternative to dynamic hedging. \\[0.5em]
%%The existing literature addressing the issue on how to generate meaningful linear (i.e. risk neutral) or nonlinear (e.g. bid and ask) prices for portfolios of derivatives in incomplete markets is substantial, both in a static and dynamic setting. We cite, among the most well known examples, the standard practice of calibrating the risk neutral measure to the option price surface (whose theoretical grounds are formalized in \cite{JacodProtter}), the theories of good deals bounds (\cite{Ledoit}, \cite{Cochrane}), indifference pricing (see the textbook \cite{Carmona} and the references therein) and Conic Finance (\cite{MadanCherny}, \cite{MadanPistoriusStadje}). Such theories come in handy when, for instance, a trader needs to quickly generate a quote for an out of the counter structured product that fits a specific client's requests. Once the price is set and the trade is concluded, however, it remains the problem of optimally hedging the risk(s) associated to such product. The main object of this work is to propose a formalization of such a problem, and to provide closed-form and numerical solutions to it.
The static hedging portfolio depends on the available derivatives and on the investor's risk exposure and risk preferences. In their paper ``Optimal positioning in derivative securities,'' 
\cite{carr-madan} consider a two period financial market, in which the world of traded assets consists of cash\footnote{In fact, \cite{carr-madan} consider zero-coupon bonds instead of cash; we use cash here for simplicity.},
a single risky asset $X=(X_t)_{t \in \{0,T\}}$, as well as calls $C(K) = (C_t(K))_{t \in \{0,T\}}$ and/or puts $P_t(K) = (P_t(K))_{t \in \{0,T\}}$ written on $X$ at every strike $K \in \Rb_+$.  In this setting, they pose the question: \textit{given an investor with initial wealth $c$, subjective/statistical probability measure $\Pb$, and utility function $U$, what is the European derivative written on $X$ that, if purchased at time zero, maximizes the investor's expected utility}?  \\[0.5em]
Mathematically, this question is formulated as follows:
\begin{align}
	&\sup_{f \in \Ac} \Eb U(f(X_T)) , &
\Ac
	&:=	\{ f : \Ebt f(X_T) \leq c \} , \label{eq:problem-1}
\end{align}
where $\Eb$ and $\Ebt$ denote expectation under the measure $\Pb$ and the risk-neutral (i.e., pricing) measure $\Pbt$, respectively.  Assuming it exists, the density of $X$ under $\Pbt$, denoted by $\pt_X$, is observable from call and/or put prices and is given by
\begin{align}
\pt_X(K)
	&=	\d_K^2 C_0(K) = \d_K^2 P_0(K) , \label{eq:pt}
\end{align}
The function $f^*$ that maximizes the supremum in \eqref{eq:problem-1} is given by
\begin{align}
f^*(x)
	&=	[U']^{-1} \Big( - \lam^* \frac{\pt_X(x)}{p_X(x)}\Big) , &
c
	&=	\int  \dd x \, \pt_X(x) f^*(x)  .
\end{align}
where $[U']^{-1}$ denotes the inverse of the derivative of $U$ and $\lambda^*$ is the Lagrange multiplier. While a European claim with payoff $f^*(X_T)$ is not traded directly, as long as $f^*$ can be written as the difference of convex functions, the payoff can be synthesized from calls and puts using 
\begin{align}
f(X_T)
	&=	f(\kappa) + f'(\kappa) \Big( C_T(\kappa) - P_T(\kappa) \Big)
			+ \int_0^\kappa f''(K) P_T(K) \dd K + \int_\kappa^\infty f''(K) C_T(K) \dd K , &
\kappa
	&\in \Rb_+ , \label{eq:replication}
\end{align}
as shown in \cite{carr-madan-2}.\footnote{Note that convex functions (and so also their differences) are almost everywhere twice continuously differentiable by Theorem 25.5. in \cite{Rockafellar} and Lebesgue's Differentiation Theorem. Furthermore, a function $f$ of a real variable can be written as the difference of convex functions on a possibly unbounded interval $I$ if and only if $f$ has left and right derivatives of bounded variation on every closed bounded interval interior to $I$ (see \cite{Hartman}). Thus, even if $f$ is only continuous, one can always at least approximate $f$ on bounded domains with the difference of convex functions (such as, e.g., a polynomial).}  Thus, the investor can obtain the claim with payoff $f^*(X_T)$ by purchasing the above basket of calls and puts.
\\[0.5em]
With a slight modification, one can transform the optimal investment problem \eqref{eq:problem-1} into an optimal hedging problem.  Suppose the investor has sold a derivative with payoff $h(X_T)$ and wishes to purchase a European claim on $X$ to maximize his expected utility.  Mathematically, this problem can be written as follows
\begin{align}
	&\sup_{f \in \Ac} \Eb U(f(X_T)-h(X_T)) , &
\Ac
	&:=	\{ f : \Ebt f(X_T) \leq c \} . \label{eq:problem-1.5}
\end{align}
It is straightforward to show that the function $f^*$ that maximizes the supremum in \eqref{eq:problem-1.5} is given by
\begin{align}
f^*(x)
	&=	[U']^{-1} \Big( - \lam^* \frac{\pt_X(x)}{p_X(x)}\Big) + h(x) , &
c
	&=	\int f^*(x) \pt(x) \dd x .
\end{align}
In their paper ``Optimal static quadratic hedging,'' \cite{leung-lorig} solve problem \eqref{eq:problem-1.5} with the utility function $U$ replaced by a concave quadratic function $U(x) = - \tfrac{1}{2} x^2$ (i.e., they minimize the $L^2(\Pb)$ norm of $f(X_T) - h(X_T)$ subject to a cost constraint). When the value of the derivative to be hedged is less than the cost constraint, i.e., $\Ebt h(X_T) \leq c$, then $f^* = h$.  However, when the cost constraint is less than the value of the option to be hedged, i.e., $c < \Ebt h(X_T)$, the solution is more involved.
\\[0.5em]
The market considered in \cite{carr-madan} and \cite{leung-lorig} is complete in the sense that every European derivative written on $X$ can be replicated using \eqref{eq:replication}.  The purpose of this paper is to extend their results   to incomplete markets.  Specifically, we will consider the two problems described below.
\\[0.5em]
\textsc{Problem 1}: \textit{Hedging non-traded risks with vanilla options}. In Section \ref{sec:non-traded}, we consider the following optimization problem
\begin{align}
	&\sup_{f \in \Ac} \Eb U(f(X_T)-h(X_T,Y_T)) , &
\Ac
	&:=	\{ f : \Ebt f(X_T) \leq c \} . \label{eq:problem-2}
\end{align}
Problem 1, stated in \eqref{eq:problem-2}, describes an investor that has sold a claim with payoff $h(X_T,Y_T)$ and now wishes to purchase a claim with payoff $f(X_T)$ with a cost of $c$ or less in order to maximize his expected utility. As such, this is a constrained optimization problem. The value of the budget constraint is left unspecified, but one can think of it as the price at which the claim with payoff $h(X_T,Y_T)$ was sold. In this scenario, the investor is using $c$, the proceeds from the sale of the claim, to fund the corresponding static hedge.\\[0.5em]
The process $Y = (Y_t)_{t \in \{0,T\}}$ in \eqref{eq:problem-2} represents a non-traded risk.  As such, payoffs of the form $h(X_T,Y_T)$ cannot be replicated in general, and the market is incomplete.  The process $X$ represents the value of an asset on which calls and/or puts are traded at every strike $K \in \Rb_+$.  As such, the density $\pt_X$ of $X_T$ under $\Pbt$ is known from \eqref{eq:pt}.  Moreover, European payoffs of the form $f(X_T)$ can be synthesized from calls and puts using \eqref{eq:replication} as long as $f$ can be written as the difference of convex functions. \\[0.5em]
%%\blu{(Note: could add example of LETF here: $L_T = X_T^\beta \exp (\tfrac{1}{2} \beta (1-\beta) \< \log X \>_T) = h(X_T,Y_T)$ where $Y_T = \< \log X \>_T$. Consider a bankder that has sold an option on $L$ and wishes to statically hedge his position with a European option on $X$.)} 
\textsc{Problem 2}: \textit{Hedging basket options with vanilla options}. In Section \ref{sec:basket}, we consider the following optimization problem
\begin{align}
	&\sup_{f,g \in \Ac} \Eb U(f(X_T) + g(Y_T) - h(X_T,Y_T)) , &
\Ac
	&:=	\{ f,g : \Ebt f(X_T) + \Ebt g(Y_T) \leq c \} . \label{eq:problem-3}
\end{align}
Problem 2, stated in \eqref{eq:problem-3}, describes an investor that has sold a claim with payoff $h(X_T,Y_T)$ and now wishes to purchase two claims with payoffs $f(X_T)$ and $g(Y_T)$, respectively, and at a total cost no higher than $c$, in order to maximize his expected utility.  The processes $X$ and $Y$ in \eqref{eq:problem-3} represent the values of assets on which calls and/or puts are traded at every strike $K \in \Rb_+$.  As such, the densities $\pt_X$ and $\pt_Y$ of $X_T$ and $Y_T$ under $\Pbt$ are known from \eqref{eq:pt}.  Moreover, European payoffs of the form $f(X_T)$ and $g(Y_T)$ can be synthesized from calls and puts using \eqref{eq:replication} as long as $f$ and $g$ can be written as the difference of convex functions.  Note that only the marginal densities $\pt_X$ and $\pt_Y$ are assumed to be observable, while the joint density $\pt_{X,Y}$ is not.  Moreover, payoffs of the form $h(X_T,Y_T)$ cannot be replicated in general.  As such, the market is incomplete.  Within the class of problems described by \eqref{eq:problem-3}, an important example is that in which $X$ and $Y$ represent the value of two stocks, a banker has sold a basket option with payoff $h(X_T,Y_T)$, and he now wishes to hedge his exposure by purchasing a European claim on $X_T$ and a European claim on $Y_T$.
\\[0.5em]
Hedging in incomplete markets also motivates us to consider a nonlinear pricing rule that accounts for the investor's risk preferences. To that end, we connect our static hedging framework with the notion of indifference pricing. In Section \ref{sec:indiff}, we define and derive a representation of the static hedging indifference pricing, along with  an illustrative example. Lastly, we include some concluding remarks in Section \ref{sec:conclude}.

%Before proceeding to solve problems \eqref{eq:problem-2} and \eqref{eq:problem-3}, let us state a few assumptions, which shall hold throughout this paper.
%\begin{itemize*}
%\item The investor's subjective probability measure and the market's chosen pricing measure are equivalent $\Pb \sim \Pbt$.
%\item The random variables $X_T$ and $Y_T$ are jointly continuous, and thus the joint probability density function $p_{X,Y}$ exists.
%\item If $X$ and/or $Y$ are traded assets, then they are $\Pbt$-martingales: $\Ebt X_T = X_0$ and/or $\Ebt Y_T = Y_0$.
%\end{itemize*}
%Finally, we wish to note that, although we in mind that $X$ and $Y$ are scalar processes, all of the results below hold if $X \in \Rb^n$, $Y \in \Rb^m$, and $h: \Rb^n \times \Rb^m \to \Rb$.

\section{Hedging non-traded risks with vanilla options}
\label{sec:non-traded}
%We consider an investor that invests in a static position in European options and wishes to maximize his expected utility from terminal wealth given he has sold a derivative with payoff $h(X_T,Y_T)$.  That is, the investor seeks to solve the following optimization problem
%\begin{align}
	%&\sup_{f \in \Ac} \Eb U(f(X_T)-h(X_T,Y_T)) , &
%\Ac	
	%&:= \{ f : \Rb \to \Rb \text{ such that } \Ebt f(X_T) \leq c \} , \label{eq:problem-2}
%\end{align}
%where $\Eb$ and $\Ebt$ denote, respectively, expectation under the real-world probability measure $\Pb$ and the pricing measure $\Pbt$.  Here, we assume $X = (X_t^{(1)},\ldots,X_t^{(n)})_{0 \leq t \leq T}$ and $Y = (Y_t^{(1)},\ldots,Y_t^{(m)})_{0 \leq t \leq T}$.  The set $\Ac$ represents all functions $f$ for which the price of a European function with payoff $f(X_T)$ is less than or equal to a constant $c$.
%\\[0.5em]

 In this section, we analyze Problem 2, stated in \eqref{eq:problem-2}, (henceforth, ``problem \eqref{eq:problem-2}'') when $U$ is an exponential, power/logarithmic, and quadratic utility function. Our results will be illustrated analytically and numerically in an example of hedging an option on a leveraged ETF.

\begin{assumption}\label{assume1} The following assumptions shall hold throughout this paper.
\begin{itemize*}
\item The subjective/statistical measure and the market's chosen pricing measure are equivalent (i.e., $\Pb \sim \Pbt$).
\item The random variables $X_T$ and $Y_T$ are jointly continuous, and, thus, the joint density $p_{X,Y}$ exists.
\item If $X$ and/or $Y$ are traded assets, then they are $\Pbt$-martingales: $\Ebt X_T = X_0$ and/or $\Ebt Y_T = Y_0$.
\end{itemize*}
\end{assumption}
\begin{remark} The processes $X$ and $Y$ are assumed to be scalar for simplicity throughout the paper. However, all the results of this section hold if $X \in \Rb^n$, $Y \in \Rb^m$, and $h: \Rb^n \times \Rb^m \to \Rb$.
\end{remark}

\noindent
To begin, let us define the Lagrangian $L$ associated with problem \eqref{eq:problem-2}.  We have
\begin{align}
L(\lam,f)
	&:=	\Eb U(f(X_T)-h(X_T,Y_T)) + \lam (\Ebt f(X_T) - c) \\
	&=	\int \dd x \int \dd y \, p_{X,Y}(x,y) U \big( f(x)-h(x,y) \big) + \lam \Big( \int \dd x \, \pt_X(x) f(x) - c \Big) . \label{eq:L1}
\end{align}
The first order optimality conditions associated with \eqref{eq:L1} are obtained by equating to zero the functional derivative of the Lagrangian $L$. Namely, given a suitable set $\mathscr{Q}$ of test functions, the following conditions must hold for all $q\in \mathscr{Q}$:
\begin{align}
0
	&=	\frac{ \del L }{\del f}(\lam^*,f^*)
	 =	\frac{\dd }{\dd \eps} L( \lam^* , f^* + \eps q) \Big|_{\eps = 0}  \\
	&=	\int \dd x \, q(x) \Big( \int \dd y \, p_{X,Y}(x,y) U' \big(f^*(x) - h(x,y) \big) + \lam^* \pt_{X}(x) \Big)  ,  \label{eq:dLdf} \\
0
	%&=	\frac{ \d L }{ \d \lam }(\lam^*,f^*)
	&=	\lam^* \Big( \int_\Rb \pt_X(x) f^*(x) \dd x - c \Big) . \label{eq:lambda-condition}
\end{align}
As \eqref{eq:dLdf} must hold for all functions $q\in \mathscr{Q}$, the function $f^*$ must satisfy
\begin{align}
0
	&=	\int \dd y \, p_{X,Y}(x,y) U' \big(f^*(x) - h(x,y) \big) + \lam^* \pt_{X}(x) .  \label{eq:optimality-1}
\end{align}
To go further, we must specify a particular form for $U$.

\begin{theorem}[Exponential Utility]
\label{thm:exp}
Suppose that the function $U$ in \eqref{eq:problem-2} is a utility function of exponential form $U(x) = - \ee^{- \gam x}/\gam$ with $\gam > 0$.  
Then, the optimizer $f^*$ is given by
\begin{align}
f^*(x)
	&=	c + \frac{1}{\gam} \log \Big( \frac{1}{\pt_X(x)}  \int \dd y \, p_{X,Y}(x,y) \ee^{\gam h(x,y)} \Big) \\ &\quad
			- \frac{1}{\gam} \int \dd \xi \, \pt_X(\xi) \log \Big( \frac{ 1}{\pt_X(\xi)} \int \dd y \, p_{X,Y}(\xi,y) \ee^{\gam h(\xi,y)}  \Big) . \label{eq:f-star-exp}
\end{align}
\end{theorem}

\begin{proof}
With $U(x) = -\ee^{-\gam x}/\gam$ we have $U'(x) = \ee^{- \gam x}$ and equation \eqref{eq:optimality-1} becomes
\begin{align}
0
	&=	\int \dd y \, p_{X,Y}(x,y) \ee^{-\gam (f^*(x) - h(x,y))} + \lam^* \pt_{X}(x) \\
	&=	\ee^{-\gam f^*(x) } \int \dd y \, p_{X,Y}(x,y) \ee^{\gam h(x,y)} + \lam^* \pt_{X}(x) . 
\end{align}
Thus, we find that $f^*$ must satisfy
\begin{align}
f^*(x)
	&=	\frac{-1}{\gam} \log ( - \lam^*) + \frac{1}{\gam} \log \Big( \frac{1}{\pt_X(x)} \int \dd y \, p_{X,Y}(x,y) \ee^{\gam h(x,y)} \Big) . \label{eq:f-star}
\end{align}
Next, note that when $U$ is a strictly increasing utility function (as is the case here), it is always optimal for $f^*$ to satisfy $\Ebt f^*(X_T) = c$. In fact, if $\Ebt f^*(X_T) < c$, then $U(f^*(X_T)+c-\Ebt f^*(X_T))> U(f^*(X_T))$, which contradicts the optimality of $f^*$. Hence, $\Ebt f^*(X_T) = c$. Then, $\lam^* \neq 0$, and \eqref{eq:lambda-condition} becomes
\begin{align}
0
	&=	\int_\Rb \dd x \, \pt_X(x) f^*(x) - c . \label{eq:dLdlambda}
\end{align}
Inserting expression \eqref{eq:f-star} for $f^*$ into \eqref{eq:dLdlambda}, we find that $\lam^*$ must satisfy
\begin{align}
\frac{1}{\gam} \log ( - \lam^*)
	&=	- c +  \frac{1}{\gam}  \int \dd x \, \log \Big( \frac{1}{\pt_X(x)} \int \dd y \,  p_{X,Y}(x,y) \ee^{\gam h(x,y)} \Big) . \label{eq:lambda-star}
\end{align}
Finally, by inserting equation \eqref{eq:lambda-star} into \eqref{eq:f-star}, we see that $f^*$ is given by \eqref{eq:f-star-exp}, as claimed.
\end{proof}

\begin{remark} As shown in the proof of Theorem \ref{thm:exp}, the Lagrange multiplier in the case of exponential utility satisfies $\lambda^*\neq 0$, and so the following integrability condition on $h$ ensures that the optimizer $f^*$ is finite:
\begin{align*}
\int \dd x \, \log \Big( \frac{1}{\pt_X(x)} \int \dd y \,  p_{X,Y}(x,y)\ee^{\gam h(x,y)} \Big)<\infty.
\end{align*}
\end{remark}

%\begin{example}
%\blu{(Would be good to add a simple numerical example here.  We could plot $h(x,y)$ as a function of $x$ for various fixed values of $y$ and overlay the function $f^*(x)$ on the same plot.)}\\
%\end{example}

\begin{theorem}[Power/Logarithmic utility]
\label{thm:pow}
Suppose that the function $U$ in \eqref{eq:problem-2} is a utility function of power or logarithmic form.  That is, 
$U(x) = x^{1-\gam}/(1-\gam)$ if $\gam \in (0,1) \cup (0,\infty)$ and $U(x) = \log x$ if $\gam = 1$.  
For any $z > 0$, define
\begin{align}
H(x,z)
	&:=	\int \dd y \, p_{X,Y}(x,y) \ee^{z^{1/\gam}h(x,y)} , \label{eq:def-H}
\end{align}
and, for any $f>0$, define the following Laplace-like transform of $H(x, \, \cdot \, )$
\begin{align}
\Lc[ H( x , \, \cdot \, ) ](f) \equiv \Hh(x,f) 
	&:= \frac{1}{\gam \Gam(\gam)} \int_0^\infty \dd z \, \ee^{-z^{1/\gam}f} H(x,z) , \label{eq:transform}
\end{align}
where $\Gam$ is the Euler-Gamma function. Then, the optimizer $f^*$ satisfies
\begin{align}
f^*(x)
	&=	\Hh^{-1}(x,- \lam^* \pt_X(x)) , &
c
	&=	\int \dd x \, \pt_X(x) f^*(x) , \label{eq:f-star-power}
\end{align}
where $\Hh^{-1}(x, \, \cdot \, )$ denotes the inverse of $\Hh(x, \, \cdot \,)$.
\end{theorem}

\begin{proof}
With $U$ in the power or logarithmic form, we have $U'(x) = 1/x^\gam$ and \eqref{eq:optimality-1} becomes
\begin{align}
0
	&=	\int \dd y \, p_{X,Y}(x,y) \frac{1}{(f^*(x) - h(x,y))^\gam} + \lam^* \pt(x) . \label{eq:1}
\end{align}
Next, using the following identity
\begin{align}
\frac{1}{x^\gam}
	&=	\frac{1}{\gam \Gam(\gam)} \int_0^\infty \dd z \, \ee^{- z^{1/\gam}x} , &
x
	&> 0 ,
\end{align}
we can re-write \eqref{eq:1} as
\begin{align}
0
	&=	\int \dd y \, p_{X,Y}(x,y) \frac{1}{\gam \Gam(\gam)} \int_0^\infty \dd z \, \ee^{-z^{1/\gam} (f^*(x) - h(x,y) )} + \lam^* \pt_X(x) \\
	&=	\frac{1}{\gam \Gam(\gam)} \int_0^\infty \dd z \, \ee^{-z^{1/\gam}f^*(x)} \int \dd y \, p_{X,Y}(x,y) \ee^{z^{1/\gam}h(x,y)} + \lam^* \pt_X(x) 
			\label{eq:tonelli} \\
	&=	\frac{1}{\gam \Gam(\gam)} \int_0^\infty \dd z \, \ee^{-z^{1/\gam}f^*(x)} H(x,z) + \lam^* \pt_X(x)  \label{eq:2} \\
	&=	\Hh(x, f^*(x)) + \lam^* \pt_X(x) , \label{eq:2.1}
%H(x,z)
	%&:=	\int \dd y \, p_{X,Y}(x,y) \ee^{z^{1/\gam}h(x,y)} , \label{eq:def-H}
\end{align}
where \eqref{eq:tonelli} follows from Tonelli's theorem, 
\eqref{eq:2} follows from the definition \eqref{eq:def-H} of $H$,
and \eqref{eq:2.1} follows from the definition \eqref{eq:transform} of $\Hh$.
Noting from \eqref{eq:def-H} that $H$ is positive by definition,
it follows from \eqref{eq:transform} that $\Hh( x , \, \cdot \, )$ is strictly decreasing and, therefore, invertible.  
Hence, we have from \eqref{eq:2.1} that $f^*$ satisfies \eqref{eq:f-star-power}, where the second equality follows because, when $U$ is strictly increasing, $f^*$ must satisfy $\Ebt f^*(X_T) = c$.
\end{proof}

%Now, for any $H : \Rb_+ \to \Rb_+$ and $f>0$ let us define the following Laplace-like transform
%\begin{align}
%\Lc[ H ](f) \equiv \Hh(f) 
	%&:= \frac{1}{\gam \Gam(\gam)} \int_0^\infty \dd z \, \ee^{-z^{1/\gam}f} H(z) . \label{eq:transform}
%\end{align}
%Then we can re-write \eqref{eq:2} as
%\begin{align}
%\Lc[ H( x, \, \cdot \,) ](f^*(x)) \equiv \Hh(x, f^*(x)) 
	%&:= - \lam^* \pt_X(x) . \label{eq:3}
%\end{align}
%Note from \eqref{eq:def-H} that the function $H$ is strictly positive.  Hence, the function $\Hh(x,f)$ is strictly decreasing in $f$ for all $x$.  
%Thus, the function $\Hh(x , \, \cdot \, )$ is invertible, and we have from \eqref{eq:3} that
%\begin{align}
%f^*(x)
	%&=	\Hh^{-1}(x,- \lam^* \pt_X(x)) , \label{eq:f-star-power}
%\end{align}
%where $\Hh^{-1}(x, \, \cdot \, )$ denotes the inverse of $\Hh(x, \, \cdot \,)$.
%Finally, inserting \eqref{eq:f-star-power} into \eqref{eq:dLdlambda}, one can obtain (either analytically or numerically) the value of $\lam^*$.
%\end{proof}

\begin{remark} 
Power/logarithmic utility requires a positive argument, and Theorem \ref{thm:pow} requires $f^*-h>0$. By equation \eqref{eq:1} we then get that $\lambda^*>0$, and so it must be the case that $f^*>0$ by equation \eqref{eq:f-star-power}.
\end{remark}

\begin{theorem}[Mean-variance Hedging]
\label{thm:quad}
Suppose the function $U$ in problem \eqref{eq:problem-2} is a concave quadratic function $U(x) = \gam x -  \tfrac{1}{2} x^2$. Then, if $\Ebt \Eb[ h(X_T,Y_T) | X_T ] + \gam \leq c$, the optimizer $f^*$ satisfies
\begin{align}
f^*(x)
	&=	\Eb[ h(X_T,Y_T) | X_T = x] + \gam .
\end{align}
Otherwise, we have
\begin{align}
f^*(x)
	&=	\Eb[ h(X_T,Y_T) | X_T = x] + \frac{ c - \gam - \Ebt \, \Eb[ h(X_T,Y_T) | X_T] }{ \Ebt \big( \pt_X(X_T) / p_X(X_T) \big) } \frac{ \pt_X(x) }{p_X(x)} + \gam .
			\label{eq:f-star-case-2}
\end{align}
\end{theorem}

\begin{proof}
Noting that $U'(x) = \gam - x$, equation \eqref{eq:optimality-1} becomes
\begin{align}
0
	&=	\int \dd y \, p_{X,Y}(x,y) (\gam - f^*(x) + h(x,y)) + \lam^* \pt_X(x)  \\
	%&=	\int \Big( \gam - f^*(x) \int p(x,y) \dd y + \int p(x,y) h(x,y) \dd y + \lam^* \int \pt(x,y) \dd y \Big) g(x) \dd x \\
	&=	( \gam - f^*(x) ) p_X(x) + \int \dd y \, p_{X,Y}(x,y) h(x,y) + \lam^* \pt_X(x) .
\end{align}
Hence, we find that $f^*$ satisfies
\begin{align}
f^*(x)
	&=	\frac{1}{p_X(x)} \int \dd y \, p_{X,Y}(x,y) h(x,y) + \lam^* \frac{ \pt_X(x) }{p_X(x)} + \gam   \\
	&=	\Eb[ h(X_T,Y_T) | X_T = x] + \lam^* \frac{ \pt_X(x) }{p_X(x)} + \gam . \label{eq:f-star-quad}
 \end{align}
When $U$ is a concave quadratic function, we cannot be sure that the optimizer $f^*$ satisfies $\Ebt f^*(X_T) = c$.  Thus, in addition to \eqref{eq:f-star-quad}, we have from \eqref{eq:lambda-condition} that either
\begin{align}
\lam^*
	&=	0 , &
	&\text{and hence}&
f^*(x)
	&=	\Eb[ h(X_T,Y_T) | X_T = x] + \gam ,
%c
	%&=	\int \dd x \, \pt_X(x) f^*(x) 
	%=		\Ebt f^*(X_T) .
\end{align}
which holds if $\Ebt \Eb[ h(X_T,Y_T) | X_T ] + \gam \leq c$, or
\begin{align}
\lam^*
	&\neq 0 , &
	&\text{and hence}&
c
	&=	\Ebt \, \Eb[ h(X_T,Y_T) | X_T] + \lam^* \Ebt \Big( \frac{ \pt_X(X_T) }{p_X(X_T)} \Big) + \gam , 
\end{align}
from which it follows that
\begin{align}
\lam^*
	&=	\frac{ c - \gam - \Ebt \, \Eb[ h(X_T,Y_T) | X_T] }{ \Ebt \big( \pt_X(X_T) / p_X(X_T) \big) } . \label{eq:lambda-star-quad}
\end{align}
In this case, by inserting \eqref{eq:lambda-star-quad} into \eqref{eq:f-star-quad}, we find that $f^*$ is given by \eqref{eq:f-star-case-2}.
%\begin{align}
%f^*(x)
	%&=	\Eb[ h(X_T,Y_T) | X_T = x] + \frac{ c - \gam - \Ebt \, \Eb[ h(X_T,Y_T) | X_T] }{ \Ebt \big( \pt_X(X_T) / p_X(X_T) \big) } \frac{ \pt_X(x) }{p_X(x)} + \gam .
%\end{align}
\end{proof}

\begin{example}
\label{ex:LETF}
Consider a continuous-time financial market $t \in [0,T]$ in which the value of an exchange traded fund (ETF) $S = \ee^X$ has dynamics of the form
\begin{align}
\dd X_t
	&=	\mu - \tfrac{1}{2} \sig^2 \dd t + \sig \dd W_t 
	=	- \tfrac{1}{2} \sig^2 \dd t + \sig \dd \Wt_t
\end{align}
where $W$ and $\Wt$ are Brownian motions under $\Pb$ and $\Pbt$, respectively, and the volatility $\sig$ is random, independent of $X$, and unknown at time zero.  Let $L = \ee^Z$ be the value of a leveraged ETF (LETF) with leverage ratio $\beta$.  The relationship between $S$ and $L$ is $\dd L_t / L_t = \beta \dd S_t / S_t$, from which we find
\begin{align}
Z_T \equiv z(X_T,Y_T)
	&=	\beta X_T - \tfrac{1}{2} \beta (\beta - 1) Y_T , &
Y_T
	&:=	T \sig^2 ,
\end{align}
where we have assumed $X_0 = Z_0 = 0$ for simplicity.  Suppose a banker has sold an option on the LETF with payoff $g(Z_T)=g(z(X_T,Y_T))=:h(X_T,Y_T)$ and now wishes to solve problem \eqref{eq:problem-2} for various choices of $U$.  Note that 
\begin{align}
p_{X,Y}(x,y)
	&=	p_{X|Y}(x,y) p_Y(y) , &
\pt_{X,Y}(x,y)
	&=	\pt_{X|Y}(x,y) \pt_Y(y) ,
\end{align}
where $p_{X|Y}(x,y) \dd x := \Pb(X_T \in \dd x | Y_T = y )$ and likewise for $\pt_{X|Y}(x,y)$.  Thus, given $\lambda,\nu>0$, if we assume $Y_T \sim \Ec(\lam)$ under $\Pb$ and $Y_T \sim \Ec(\nu)$ under $\Pbt$, where, for any $\ell>0$, $\Ec(\ell)$ denotes an exponential distribution with parameter $\ell$, we have
\begin{align}
p_{X,Y}(x,y)
	&=	\frac{\lam \ee^{- \lam y}}{\sqrt{2 \pi y}} \exp \Big( - \frac{(x - \mu T + y/2)^2}{2y} \Big)  , &
\pt_{X,Y}(x,y)
	&=	\frac{\nu \ee^{- \nu y}}{\sqrt{2 \pi y}} \exp \Big( - \frac{(x + y/2)^2}{2y} \Big)  ,
\end{align}
where we have used the fact that $X_T | Y_T = y \sim \Nc( \mu T - y/2, y)$ under $\Pb$ and $X_T | Y_T = y \sim \Nc( - y/2, y)$ under $\Pbt$.  The (unconditional) density of $X_T$ under $\Pb$ and $\Pbt$, obtained by integrating the joint densities with respect to $y$, are
\begin{align}
p_X(x)
	%&=	\int_0^\infty \dd y \, p_{X,Y}(x,y)
	&=		\frac{2 \lam }{\sqrt{8 \lam +1}} \exp \Big( -\frac{|x-\mu T|}{2} \sqrt{8 \lam +1}  -\frac{(x-\mu T)}{2} \Big) , &
\pt_X(x)
	%&=	\int_0^\infty \dd y \, \pt_{X,Y}(x,y)
	&=		\frac{2 \nu }{\sqrt{8 \nu +1}} \exp \Big( -\frac{|x|}{2} \sqrt{8 \nu +1}  -\frac{x}{2} \Big) .
\end{align}
Given $g$, the above information is sufficient to explicitly compute $f^*$ in Theorems \ref{thm:exp} (exponential utility) and \ref{thm:quad} (mean-variance).  As the expressions for $f^*$ are quite long, we will not present them here.  In Figure \ref{fig:LETF}, we plot $f^*(\log s)$ as a function of $s$ for $\beta \in \{2,-2\}$ and $g(z)=z$.  For comparison, we also plot $h(\log s, 1/\lam)$ (note that $\Eb \sig^2 T = 1/\lam$).

\begin{figure}[!ht]
\begin{tabular}{cc}
\includegraphics[width=0.45\textwidth]{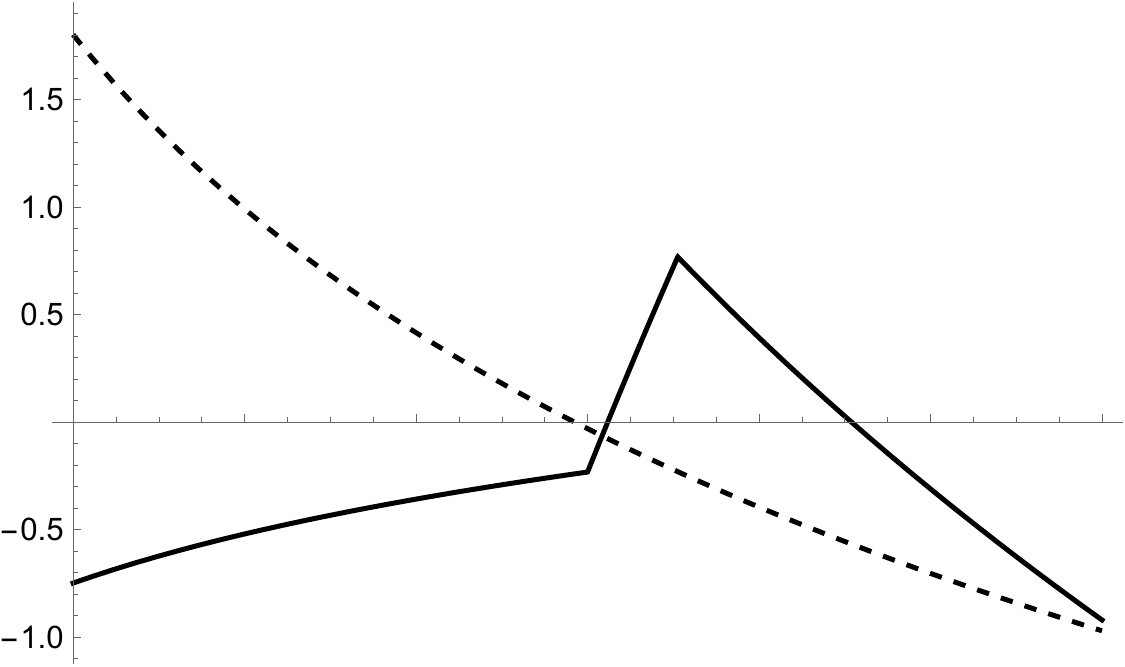}&
\includegraphics[width=0.45\textwidth]{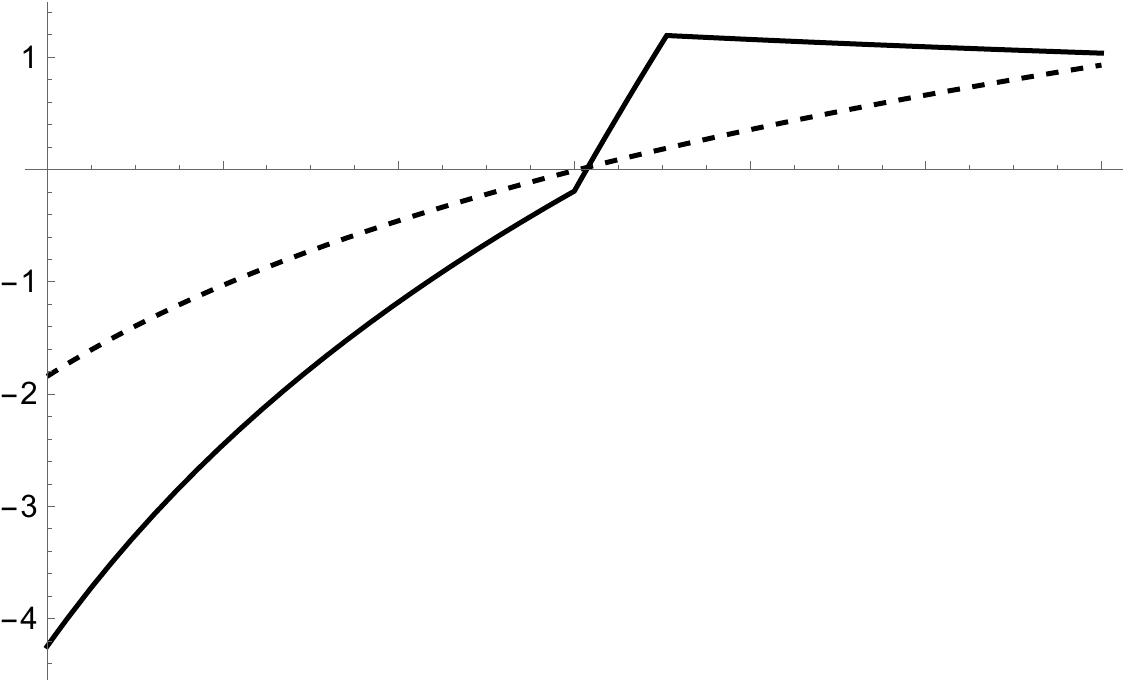}\\
$\beta = -2$, $U(x) = - \ee^{- \gam x} / \gam$ & $\beta = 2$, $U(x) = - \ee^{- \gam x} / \gam$ \\
\includegraphics[width=0.45\textwidth]{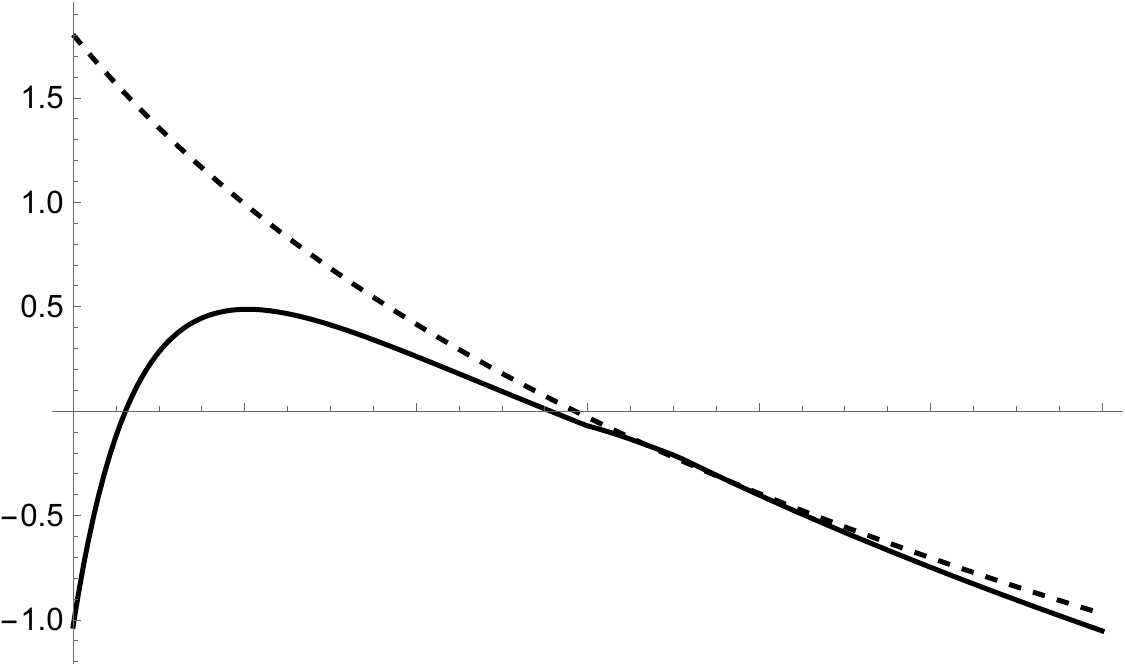}&
\includegraphics[width=0.45\textwidth]{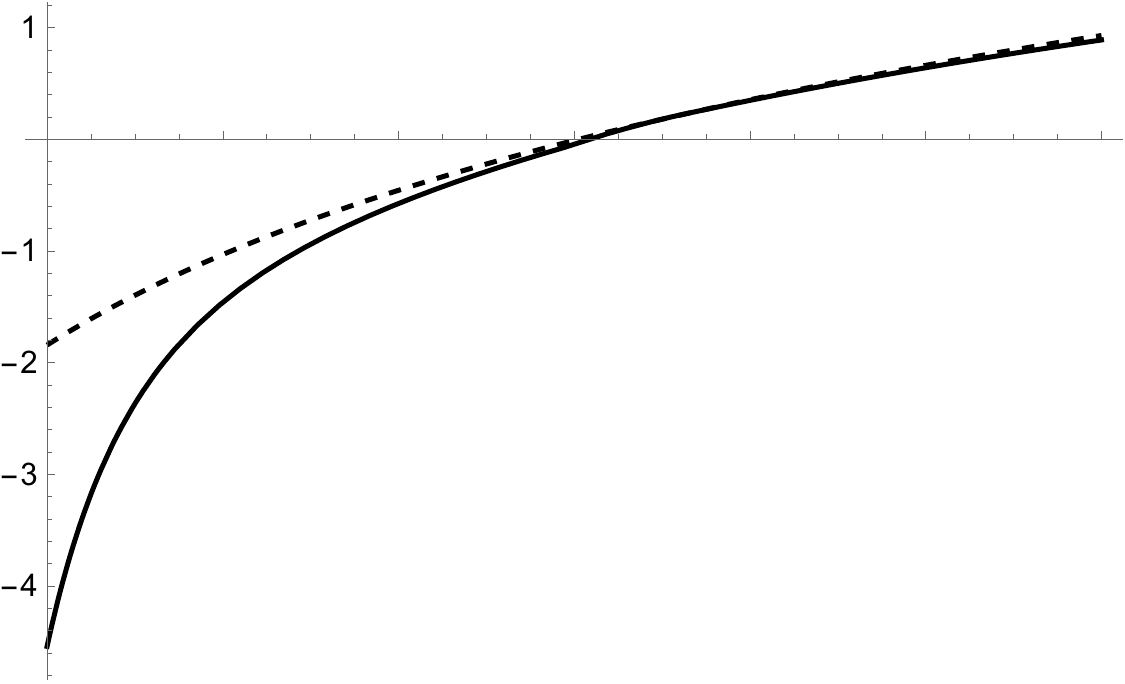}\\
$\beta = -2$, $U(x) = - \tfrac{1}{2} x^2$ & $\beta = 2$, $U(x) = - \tfrac{1}{2} x^2$
\end{tabular}
\caption{
For the model considered in Example \ref{ex:LETF}, we plot $f^*(\log s)$ (solid line) and $h(\log s,1/\lam)$ (dashed line) as a function of $s$.  Model parameters in the plot are as follows: $\lam = 1/(0.1)^2$, $\nu = 1/(0.15)^2$, $\mu = 0,1$ and $T = 1.0$.  For the case of exponential utility, we set $\gam = 2$ and the cost constraint equal to the price of the option sold $c = \Ebt h(X_T,Y_T)$.  In the case of mean-variance maximization, we set $\gam = 0$, and the cost constraint is given by
$
c
	=	\int \dd x \, \pt_X(x) \int \dd y \, \frac{ p_{X,Y}(x,y) }{ p_X(x) } h(x,y) + \gam - \Big| \int \dd x \int \dd y \, \pt_{X,Y}(x,y) h(x,y) \Big| ,
$
which guarantees that $f^*$ is given by \eqref{eq:f-star-case-2}.
}
\label{fig:LETF}
\end{figure}

\end{example}

 \section{Hedging basket options with vanilla options}
\label{sec:basket}
In this section, we solve Problem 2, stated in \eqref{eq:problem-3} (henceforth, ``problem \eqref{eq:problem-3}''), when $U$ is an exponential utility function and when $U$ is a concave quadratic utility function. In addition to deriving the optimizers, we construct for the case of exponential utility a practical iterative scheme to converge to the optimal position. We illustrate our results with several numerical examples.\\[0.5em]
 As in Section \ref{sec:non-traded}, we assume that the conditions in Assumption \ref{assume1} hold. First, let us define the Lagrangian $L$ associated with problem \eqref{eq:problem-3}.  We have
\begin{align}
L(\lam,f,g)
	&:=	\Eb U \big( f(X_T) + g(Y_T) - h(X_T,Y_T) \big) + \lam ( \Ebt f(X_T) + \Ebt g(Y_T) - c) \\
	&=	\int \dd x \int \dd y \, p_{X,Y}(x,y) U \big( f(x) + g(y) -h(x,y) \big) + \lam \Big( \int \dd x \, \pt_X(x) f(x) + \int \dd y \, \pt_Y(y) g(y) - c \Big) .
\end{align}
The first order optimality conditions are
\begin{align}
0
	&=	\frac{ \del L }{\del f}(\lam^*,f^*,g^*)
	 =	\frac{\dd }{\dd \eps} L( \lam^* , f^* + \eps q, g^*) \Big|_{\eps = 0}  \\
	&=	\int \dd x \, q(x) \Big( \int \dd y \, p_{X,Y}(x,y) U' \big(f^*(x) + g^*(y) - h(x,y) \big) + \lam^* \pt_X(x) \Big)  ,  \label{eq:dLdf-2} \\
0
	&=	\frac{ \del L }{\del g}(\lam^*,f^*,g^*)
	 =	\frac{\dd }{\dd \eps} L( \lam^* , f^* , g^* + \eps q) \Big|_{\eps = 0}  \\
	&=	\int \dd y \, q(y) \Big( \int \dd x \, p_{X,Y}(x,y) U' \big(f^*(x) + g^*(y) - h(x,y) \big) + \lam^* \pt_Y(y) \Big)  ,  \label{eq:dLdg} \\
0
	&=	\lam^* \Big( \int \dd x \, \pt_X(x) f^*(x) + \int \dd y \, \pt_Y(y) g(y) - c \Big) . \label{eq:lambda-condition-2}
\end{align}
As \eqref{eq:dLdf-2} and \eqref{eq:dLdg} must hold for all suitable functions $q$, we find that $f^*$ and $g^*$ must satisfy
\begin{align}
0
	&= \int \dd y \, p_{X,Y}(x,y) U' \big(f^*(x) + g^*(y) - h(x,y) \big) + \lam^* \pt_X(x) , \label{eq:opt-f} \\
0
	&=	\int \dd x \, p_{X,Y}(x,y) U' \big(f^*(x) + g^*(y) - h(x,y) \big) + \lam^* \pt_Y(y) . \label{eq:opt-g} 
\end{align}
%When $U$ is a concave increasing utility function, it is always optimal for $\Ebt f^*(X_T) + \Ebt g^*(Y_T) = c$.  In this case, we have $\lam^* \neq 0$ and \eqref{eq:lambda-condition-2} becomes
%\begin{align}
%0
	%&=	\int \dd x \, \pt_X(x) f^*(x) + \int \dd y \, \pt_Y(y) g^*(y)  - c . \label{eq:dLdlambda-2}
%\end{align}
To go further, we must specify the function $U$.

\begin{remark}\label{Uniq}[On the uniqueness of $(f^*,g^*)$]
If $f^*$ and $g^*$ are optimal, then, so are $f^* + k$ and $g^* - k$, where $k \in \Rb$.  As such, it is clear that we will not have unique solutions to problem \eqref{eq:problem-3}.
\end{remark}

\subsection{Exponential utility}\label{exp2risks}

\begin{theorem}[Exponential utility]
\label{thm:exp2}
Suppose that the function $U$ in \eqref{eq:problem-3} is a utility function of exponential form $U(x) = - \ee^{- \gam x}$ with $\gam > 0$. Then, the pair of optimizers $f^*$ and $g^*$ satisfy
\begin{align}
f^*(x)
	&=	\frac{-1}{\gam} \log ( - \lam^*) + \frac{1}{\gam} \log \Big( \frac{ 1 }{\pt_X(x)} \int \dd y \, p_{X,Y}(x,y) \ee^{\gam (h(x,y) - g^*(y))} \Big) , 
			\label{eq:fstar} \\
g^*(y)
	&=	\frac{-1}{\gam} \log ( - \lam^*) + \frac{1}{\gam} \log \Big( \frac{ 1 }{\pt_Y(y)} \int \dd x \, p_{X,Y}(x,y) \ee^{\gam (h(x,y) - f^*(x))} \Big) , 
			\label{eq:gstar} \\
c
	&=	\int \dd x \, p_X(x) f^*(x) + \int \dd y \, \pt_Y(y) g^*(y) . \label{eq:dLdlambda-2}
\end{align}
\end{theorem}

\begin{proof}
%In this Section, as in Section \ref{sec:exp}, we consider exponential utility $U(x) = - \ee^{- \gam x} / \gam$ so that $U'(x) = \ee^{- \gam x}$.  
By comparing \eqref{eq:optimality-1} and \eqref{eq:opt-f}, we see that $h(x,y)$ in the former has been replaced by $h(x,y) - g^*(y)$ in the latter.  As such, we find from \eqref{eq:f-star} that $f^*$ and $g^*$ must satisfy \eqref{eq:fstar} and \eqref{eq:gstar}, respectively, where the expression for $g^*$ follows by comparing
\eqref{eq:opt-f} and \eqref{eq:opt-g}.  Next, note that when $U$ is a strictly increasing utility function (as is the case here), it is always optimal for the pair $f^*$ and $g^*$ to satisfy $\Ebt f^*(X_T) + \Ebt g^*(Y_T) = c$.  In this case, we have $\lam^* \neq 0$ and \eqref{eq:lambda-condition-2} becomes \eqref{eq:dLdlambda-2}.
\end{proof}

%\begin{align}
%%0
	%%&=	\frac{\del L}{\del f}(\lam^*,f^*,g^*) &
	%%&\Rightarrow&
%f^*(x)
	%&=	\frac{-1}{\gam} \log ( - \lam^*) + \frac{1}{\gam} \log \Big( \frac{ 1 }{\pt_X(x)} \int \dd y \, p_{X,Y}(x,y) \ee^{\gam (h(x,y) - g^*(y))} \Big) , 
			%\label{eq:fstar}\\
%%0
	%%&=	\frac{\del L}{\del g}(\lam^*,f^*,g^*) &
	%%&\Rightarrow&
%g^*(y)
	%&=	\frac{-1}{\gam} \log ( - \lam^*) + \frac{1}{\gam} \log \Big( \frac{ 1 }{\pt_Y(y)} \int \dd x \, p_{X,Y}(x,y) \ee^{\gam (h(x,y) - f^*(x))} \Big) , 
			%\label{eq:gstar}
%\end{align}
%where the expression for $g^*$ follows by comparing \eqref{eq:opt-f} and \eqref{eq:opt-g}.
%\\[0.5em] 

%\begin{remark}[\red{Additional assumptions}] To apply convex optimization results I think we need a reflexive Banach space, over which $\Eb U(\cdot)$ is at least lower semicontinuous in the weak topology. The Orlicz heart generated by the Young function $\Phi(x)=e^{-\gamma x}-1$  could be the right set up.
%\end{remark}

\subsubsection{Capital allocation}

%The non-uniqueness of $f^*$ and $g^*$ create some issues, as we show next. Suppose we define
%\begin{align}
%F(x)
%	&:=	\ee^{\gam f^*(x)} , &
%G(y)
%	&:= \ee^{\gam g^*(y)} , &
%H(x,y)
%	&:=	\ee^{\gam h(x,y)} . \label{eq:FGH}
%\end{align}
%Then we have from \eqref{eq:fstar} and \eqref{eq:gstar} that 
%\begin{align}
%F(x)
%	&=	\frac{-1}{\lam^* \pt_X(x)} \int \dd y \, p_{X,Y}(x,y) \frac{H(x,y)}{G(y)} , \label{eq:F-1} \\
%G(y)
%	&=	\frac{-1}{\lam^* \pt_Y(y)} \int \dd x \, p_{X,Y}(x,y) \frac{H(x,y)}{F(x)} . \label{eq:G-1}
%\end{align}
%And, if we insert \eqref{eq:G-1} into \eqref{eq:F-1} and vice versa, we obtain
%\begin{align}
%F(x)
%	&=	\frac{1}{\pt_X(x)} \int \dd y \, p_{X,Y}(x,y) H(x,y) \pt_Y(y) \Big(
%			\int \dd \xi \, p_{X,Y}(\xi,y) \frac{ H(\xi,y) }{ F(\xi) }
%			\Big)^{-1} , \label{eq:F} \\
%G(y)
%	&=	\frac{1}{\pt_Y(y)} \int \dd x \, p_{X,Y}(x,y) H(x,y) \pt_X(x) \Big(
%			\int \dd \eta \, p_{X,Y}(x,\eta) \frac{H(x,\eta) }{ G(\eta) }
%			\Big)^{-1} . \label{eq:G}
%\end{align}
%But, now it looks like $F$ and $G$ (and therefore $f^*$ and $g^*$) do not depend on $\lam^*$.  
%
%Also, if we insert \eqref{eq:fstar} and \eqref{eq:gstar} into \eqref{eq:dLdlambda-2} we find that $\lam^*$ should satisfy
%\begin{align}
%\frac{2}{\gam} \log ( - \lam^*)
%	&=	- c + \frac{1}{\gam}  \int \dd x \, \pt_X(x)  \log \Big( \frac{ 1 }{\pt_X(x)} \int \dd y \, p_{X,Y}(x,y) \ee^{\gam (h(x,y)-g^*(y))} \Big) \\ &\quad
%			+ \frac{1}{\gam}  \int \dd y \, \pt_Y(y)  \log \Big( \frac{ 1 }{\pt_Y(y)} \int \dd x \, p_{X,Y}(x,y) \ee^{\gam (h(x,y)-f^*(x))} \Big) .
%\end{align}

The solution identified by equations \eqref{eq:fstar}-\eqref{eq:dLdlambda-2} is such that the portions of the total budget $c$ allocated to hedging against the risks $X_T$ and $Y_T$ (i.e., $\Ebt f^*(X_T)$ and $\Ebt g^*(Y_T)$) are respectively given by
\begin{align*}
\Ebt f^*(X_T)
& = \frac{-1}{\gam} \log ( - \lam^*) 
	+ \frac{1}{\gam} \int \dd x \, \pt_X(x)  
	  \log \Big( \frac{ 1 }{\pt_X(x)}
	  \int \dd yp_{X,Y}(x,y) 
	  \ee^{\gam (h(x,y)-g^*(y))} \Big), \\
\Ebt g^*(Y_T)
& = \frac{-1}{\gam} \log ( - \lam^*) 
	+ \frac{1}{\gam} \int \dd y \, \pt_Y(y)  
	  \log \Big( \frac{ 1 }{\pt_Y(y)} 
	  \int \dd x \, p_{X,Y}(x,y)
	  \ee^{\gam (h(x,y)-f^*(x))} \Big).
\end{align*}
Furthermore, if we insert \eqref{eq:fstar} and \eqref{eq:gstar} into \eqref{eq:dLdlambda-2}, we find that $\lam^*$ satisfies
\begin{align}\label{eq:lam}
\frac{-1}{\gam} \log ( - \lam^*)
	&=	\frac{c}{2} - \frac{1}{2\gam}  \int \dd x \, \pt_X(x)  \log \Big( \frac{ 1 }{\pt_X(x)} \int \dd y \, p_{X,Y}(x,y) \ee^{\gam (h(x,y)-g^*(y))} \Big) \\ &\quad
			- \frac{1}{2\gam}  \int \dd y \, \pt_Y(y)  \log \Big( \frac{ 1 }{\pt_Y(y)} \int \dd x \, p_{X,Y}(x,y) \ee^{\gam (h(x,y)-f^*(x))} \Big),
\end{align}
and so the solution $(f^*,g^*)$ found in Theorem \ref{thm:exp2} also solves
\begin{align}
\Ebt f^*(X_T)
	&=\frac{c}{2}
		+\frac{1}{2\gamma}\int \dd \xi \, \pt_X(\xi)
			\log\left(\frac{ 1 }{\pt_X(\xi)}
		 	\int \dd y \, p_{X,Y}(\xi,y)\ee^{\gamma (h(\xi,y)
			-g^*(y))}\right)\\
	& \ \ \
		-\frac{1}{2\gamma}\int \dd \xi \, \pt_Y(\xi)
			\log\left(\frac{ 1 }{\pt_Y(\xi)}
		 	\int \dd x \, p_{X,Y}(x,\xi)\ee^{\gamma (h(x,\xi)
			-f^*(x))}\right), \label{eq:Ebtf}\\
\Ebt g^*(Y_T)
	&=\frac{c}{2}
		+\frac{1}{2\gamma}\int \dd \xi \, \pt_Y(\xi)
			\log\left(\frac{ 1 }{\pt_Y(\xi)}
			\int \dd x \, p_{X,Y}(x,\xi)\ee^{\gamma (h(x,\xi)
			-f^*(x))}\right)\\
	& \ \ \
		-\frac{1}{2\gamma}\int \dd \xi \, \pt_X(\xi)
			\log\left(\frac{ 1 }{\pt_X(\xi)}
		 	\int \dd y \, p_{X,Y}(\xi,y)\ee^{\gamma (h(\xi,y)
			-g^*(y))}\right). \label{eq:Ebtg}
\end{align}
Putting together equations \eqref{eq:fstar}, \eqref{eq:gstar}, \eqref{eq:lam}, \eqref{eq:Ebtf} and \eqref{eq:Ebtg}, we obtain
%
%, note that, as in standard multivariate optimization, if $(f^*,g^*)$ is optimal, $f^*$ must also be optimal for hedging $h(X_T,Y_T)-g^*(Y_T)$ with budget $c-\Ebt[g^*(Y_T)]$. 
%In fact: 
%\begin{align*}
%\Eb[U(g^*(Y_T)+f^*(X_T)-h(X_T,Y_T))]
%&\leq \sup_{\Ebt[f(X_T)+g^*(Y_T)]\leq c}
%	\Eb[U(g^*(Y_T)+f(X_T)-h(X_T,Y_T)]\\
%&\leq \sup_{\Ebt[f(X_T)+g(Y_T)]\leq c}
%	\Eb[U(g(Y_T)+f(X_T)-h(X_T,Y_T)]\\
%&=\Eb[U(g^*(Y_T)+f^*(X_T)-h(X_T,Y_T)].
%\end{align*}
%By the results above,
%\begin{align*}
%f^*(x) = c-\Ebt[g(Y_T)]
%		+\frac{1}{\gamma}\log\left(\frac{ 1 }{\pt_X(x)}
%		\int \dd y \, p_{X,Y}(x,y)
%		\ee^{\gamma (h(x,y)-g^*(y))}\right)
%\end{align*} 
%Similarly for $g^*$. The first order conditions then become:
\begin{align}
f^*(x) &= c_f
		+\frac{1}{\gamma}\log\left(\frac{ 1 }{\pt_X(x)}
		\int \dd y \, p_{X,Y}(x,y)\ee^{\gamma (h(x,y)
			-g^*(y))}\right)\\
		&\quad 
		-\frac{1}{\gamma}\int \dd \xi \, \pt_X(\xi)
			 \log\left(\frac{ 1 }{\pt_X(\xi)}
		 	\int \dd y \, p_{X,Y}(\xi,y)\ee^{\gamma (h(\xi,y)
			-g^*(y))}\right)\label{eq:fstar1}\\
g^*(y) &= c_g
		+\frac{1}{\gamma}\log\left(\frac{ 1 }{\pt_Y(y)}
		\int \dd x \, p_{X,Y}(x,y)\ee^{\gamma (h(x,y)
			-f^*(x))}\right)\\
		&\quad
		-\frac{1}{\gamma}\int \dd \xi \, \pt_Y(\xi)
			 \log\left(\frac{ 1 }{\pt_Y(\xi)}
		 	\int \dd x \, p_{X,Y}(x,\xi)\ee^{\gamma (h(x,\xi)
			-f^*(x))}\right), \label{eq:gstar1}
\end{align}
where
\begin{align*}
c_f : = c-\Ebt g^*(Y_T), \ c_g : =  c-\Ebt f^*(Y_T).
\end{align*}
On the other hand, simple algebra shows that each solution of the system defined by equations \eqref{eq:fstar1} and\eqref{eq:gstar1} is also a solution to \eqref{eq:fstar} and \eqref{eq:gstar}. This is not a surprise since, as in standard multivariate optimization, if $(f^*,g^*)$ is optimal, $f^*$ must also be optimal for hedging $h(X_T,Y_T)-g^*(Y_T)$ with budget $c-\Ebt[g^*(Y_T)]$, and similarly for $g^*$. In particular, the same line of reasoning as in Theorem \ref{thm:exp} shows that \eqref{eq:fstar1} and \eqref{eq:gstar1} are the first order conditions for the problem
\begin{align}
	&\sup_{f,g \in \Ac'} \Eb U(f(X_T) + g(Y_T) - h(X_T,Y_T)) , &
\Ac'
	&:=	\{ f,g : \Ebt f(X_T) = c_f, \ \Ebt g(Y_T) = c_g, \ c_f+c_g = c \}, \label{eq:problem-3'}
\end{align}
Under some additional conditions (e.g., as specified below, $\Omega$ finite), we will show that $f^*+g^*$ is unique, even though $(f^*,g^*)$ is not unique, and it can be constructed by iterating a certain operator $\mathscr{H}$ (see \eqref{eq:H}). The hedges $f^*$ and $g^*$ will then be constructed from $f^*+g^*$ based on equations \eqref{eq:fstar1} and \eqref{eq:gstar1} (see, in particular, equations \eqref{eq:fstarN} and \eqref{eq:gstarN}).

\subsubsection{Construction of the iterative scheme}
It is easy to see from the formulation \eqref{eq:problem-3'} that, in the case of exponential utility, the function $f^*+g^*$, where $f^*$ and $g^*$ solve \eqref{eq:fstar} and \eqref{eq:gstar}, can be constructed as $\eta^*+c$, where $\eta^*$ solves the problem
\begin{align}
	&\sup_{\eta \in \Ac''} \Eb U(\eta(X_T,Y_T) - h(X_T,Y_T)+c) , &
\Ac''
	&:=	span\{\sigma(X_T),\sigma(Y_T)\}\cap
		\{ \eta : \Ebt \eta(X_T,Y_T) = 0\}. \label{eq:problem-3_entropic}
\end{align}
Consider the operator 
\begin{align*}
\Hc:\sigma(X_T,Y_T)\rightarrow \sigma(X_T,Y_T)
\end{align*}
defined by 
\begin{align}\label{eq:H}
\Hc(\xi):=\Hc^Y(\Hc^X(\xi)), \ \xi\in \sigma(X_T,Y_T),
\end{align}
where
\begin{align}
\Hc^X(\xi) &:=
		 \xi
		-\left(\Ec^X(\xi)-\Ebt\Ec^X\xi\right),
		\label{eq:HX}\\
\Hc^Y(\xi) &:= 
		\xi
		-\left(\Ec^Y(\xi)-\Ebt\Ec^Y\xi\right),
		\label{eq:HY}\\
\Ec^X(\xi) &:=
		\frac{1}{\gamma}\log\left(\frac{ 1 }{\pt_X(x)}
		\int \dd y \, p_{X,Y}(x,y)\ee^{\gamma \xi(x,y)
			}\right) ,\label{eq:epsX}\\
\Ec^Y(\xi) &:= 
		\frac{1}{\gamma}\log\left(\frac{ 1 }{\pt_Y(y)}
		\int \dd x \, p_{X,Y}(x,y)\ee^{\gamma \xi(x,y)
			}\right).\label{eq:epsY}			
\end{align}
The rest of Section \ref{exp2risks} is devoted to showing that the solution $\eta^*$ to problem \eqref{eq:problem-3_entropic} satisfies $\eta^*=h-h^*-c$, where $h^*$ is the unique fixed point of the operator $\mathscr{H}$ on the set $B_{\Eb U(h-c)}$ and where, for $R\geq 0$, the set $B_R$ is defined by
\begin{align}
B_{R} := \{\xi:\Eb U(-\xi)\geq R\}\cap\{\xi:\Ebt \xi=0\}	
		\cap \{\xi:h-\xi-c\in \Ac''\},\label{eq:Domain}
\end{align}
where $\Ac''$ is as in \eqref{eq:problem-3_entropic}. 
\begin{remark}\label{RNextension}
The definition of the operator $\mathscr{H}$ requires one to extend the expectation $\Ebt$ to the $\sigma$-algebra $\sigma(X_T,Y_T)$, in such a way that it corresponds to the expectation with respect to a measure $\Pbt_{X_T,Y_T}$ whose marginals for $X_T$ and $Y_T$ are the same as those under the risk neutral distribution $\Pbt$. This can be done by simply taking the joint density of $(X_T,Y_T)$ under $\Pbt_{X_T,Y_T}$ as the product of the (assumed known) risk neutral marginals.
\end{remark}
\begin{remark} In theory, one also needs $\Ebt[h(X_T,Y_T)]=c$, the budget available to hedge against the risks $X_T$ and $Y_T$. However, because $U$ is the exponential utility, any fund available in excess or in defect of  $\Ebt[h(X_T,Y_T)]$ may simply be added to or subtracted by the solution found through the iterative scheme proposed below. Consequently, we simply assume in what follows that, indeed, $\Ebt[h(X_T,Y_T)]=c$.
\end{remark} 

\begin{proposition} The operator $\Hc$, defined by \eqref{eq:H}, maps the set $B_R$ defined by \eqref{eq:Domain} into itself for any constant $R\geq 0$. Furthermore, it is Lipschitz continuous as a map 
\begin{align*}
\Hc:L^{\infty}(\Omega,\sigma(X_T,Y_T),\Pb)\rightarrow L^{\infty}(\Omega,\sigma(X_T,Y_T),\Pb).
\end{align*}
\end{proposition}
\begin{proof}
Fix $R\geq 0$. Clearly, $\Ebt \Hc^X(\xi)=0$ if $\xi\in B_R$. Suppose $h-\xi-c=f+g$ for some $f\in\sigma(X_T)$ and $g\in \sigma(Y_T)$ that satisfy 
\begin{align*}
\Ebt f(X_T)+\Ebt g(Y_T)=0.
\end{align*}
Then,
\begin{align*}
\Hc^X(\xi)
=h-c-f-g-\left(\Ec^X(h-c-g-f)-\Ebt \Ec^X(h-c-g-f)\right),
\end{align*}
and $\Hc^X(\xi)-h+c\in \Ac''$. If $\xi\in B_{R}$, Theorem \ref{thm:exp} implies
\begin{align*}
R\leq \Eb U(-\xi)\leq \Eb U(-\Hc^X(\xi)),
\end{align*}
and so $\Hc^X(\xi)\in B_{R}$. Applying the same argument to $\Hc^Y(\Hc^X(\xi))$ then yields $\Hc(\xi)\in B_R$. As for the Lipschitz continuity, suppose that $\xi_1,\xi_2\in L^{\infty}(\Omega,\sigma(X_T,Y_T),\Pb)$. Then, the triangle inequality, monotonicity and cash invariance of the entropic risk measure (see \cite{FollmerSchied}), and the absolute value inequality for expectation give
\begin{align*}
\|\Hc(\xi_1)-\Hc(\xi_2)\|_{\infty}
&\leq 
	\|\Hc^X(\xi_1)-\Hc^X(\xi_2)\|_{\infty}\\
& \ \ \
	+\|\Ec^Y(\Hc^X(\xi_1))-\Ebt\Ec^Y(\Hc^X(\xi_1)
		-\Ec^Y(\Hc^X(\xi_2))+\Ebt\Ec^Y(\Hc^X(\xi_2)\|_{\infty}
		\\
&\leq
	\|\Hc^X(\xi_1)-\Hc^X(\xi_2)\|_{\infty}\\
& \ \ \
	+\|\Hc^X(\xi_1)-\Hc^X(\xi_2)\|_{\infty}
	+\Ebt|\Ec^Y(\Hc^X(\xi_1))-\Ec^Y(\Hc^X(\xi_2))|\\
&\leq
	3\|\Hc^X(\xi_1)-\Hc^X(\xi_2)\|_{\infty}\\
&\leq
	9\|\xi_1-\xi_2\|_{\infty}.
\end{align*}
\end{proof}
\begin{proposition} The set $B_{R}$ is closed and convex for every $R\geq 0$.
\end{proposition}
\begin{proof}
This follows immediately on noting that, for every $R\geq 0$, $B_{R}$ is the intersection of spaces that are closed and convex. 
\end{proof}

\subsubsection{Finite state space and a fixed point scheme}

At this point, additional assumptions on the geometry of the unit ball are needed. For simplicity, it is assumed for the rest of Section \ref{exp2risks} that there are only finitely many states of the world, each occurring with positive probability:

\begin{assumption}\label{assumption2}
The cardinality of the set $\Omega$ is $|\Omega|=N\in \Nb$ and $\Pb(\omega)>0 \ \forall\omega\in \Omega$. 
\end{assumption}

\noindent
See Remark  \ref{rmk:reflexive} for potential extensions of Assumption \ref{assumption2}. We also recall the definition of strict concavity, and a well-known fundamental result concerning the uniqueness of solutions to convex maximization problems with strictly concave targets.

\begin{definition}
A function $u:\Rb^n\rightarrow \Rb$. $n\in \Nb$, is said to be strictly concave if 
\begin{align*}
f(\tau x+(1-\tau)y)>\tau f(x)+(1-\tau)f(y)
\end{align*}
for every $x,y\in \Rb^n$ and $\tau\in (0,1)$.
\end{definition}

\begin{theorem}
Suppose that $B\subset \Rb^n$ is convex, and let $u$ be a real-valued and strictly concave function on $B$. Then, $u$ admits at most one maximum on $B$.
\end{theorem}
\begin{proof}
Suppose that $u$ admits two distinct maxima on $B$, denoted by $x^*$ and $y^*$. Since $B$ is convex, for every $\tau\in (0,1)$, $\tau x^*+(1-\tau)y^*\in B$. But, since $u$ is strictly concave,
\begin{align*}
u(\tau x^*+(1-\tau)y^*)>\tau u(x^*)+(1-\tau)u(y^*) = u(x^*),
\end{align*}
which contradicts the maximality of $x^*$.
\end{proof}

\noindent
Under Assumption \ref{assumption2}, problem \eqref{eq:problem-3_entropic} is also finite dimensional. In particular, given $X_T$ and $Y_T$, the sets $\Ac''$ and $B_R$ defined in \eqref{eq:problem-3_entropic} and \eqref{eq:Domain} for $R\geq 0$ can be identified with subsets of $\Rb^N$. Then, the map
\begin{align}
\mathscr{U}:\eta\rightarrow \Eb U(\eta(X_T,Y_T)-h(X_T,Y_T)+c),\label{eq:EU}
\end{align}
is a real valued function on $\mathscr{A}''\subset \Rb^N$. It is strictly concave and continuous, since $U$ is strictly concave and continuous and the expectation is a linear operator. Next, we show some useful lemmas.

\begin{lemma} The set $B_{R}$ is bounded.
\end{lemma}
\begin{proof}
Suppose $B_{R}$ is unbounded. Then, there is $\{\xi_n\}_{n\in\Nb}\subset B_{R}$ such that $\|\xi_n\|_{\infty}\rightarrow \infty$. Since $B_{R}$ is convex, 
\begin{align*}
L_n:=\left\lbrace
	\frac{\lambda}{\|\xi_n\|_\infty} \xi_n
	+\left(1-\frac{\lambda}{\|\xi_n\|_\infty}\right)\xi_1
	\right\rbrace_{\lambda\in [0,\|\xi_n\|_{\infty}]}\subset B_{R}
\end{align*}
for every $n\in \Nb$. Since the sequence $\{\xi_n/\|\xi_n\|_{\infty}\}_{n\in \Nb}$ lies on a compact set, it admits a convergent subsequence $\{\xi_{n_k}/\|\xi_{n_k}\|_{\infty}\}_{k\in \Nb}$. Let $\xi$ be its limit point. Since $\|\xi_n\|_{\infty}\rightarrow \infty$, for every $\lambda \geq 0$ there is $\hat{k}\in\Nb$ such that for every $k\in \Nb$, $k\geq \hat{k}$,
\begin{align*}
\frac{\lambda}{\|\xi_{n_k}\|_\infty} \xi_{n_k}
	+\left(1-\frac{\lambda}{\|\xi_{n_k}\|_\infty}\right)\xi_1 \in B_{R}.
\end{align*}
Letting $k\rightarrow \infty$ and since $B_R$ is closed, it must be the case that $\lambda\xi +\xi_1\in B_R$. Hence, $\{\lambda \xi+\xi_1\}_{\lambda\geq 0}\subset B_{R}$. Since $\Ebt\xi_n=0$, it must be the case that 
\begin{align*}
\lim_{n\rightarrow \infty} \sup_{\omega \in \Omega} \xi_n(\omega) = +\infty,
\end{align*}
and so there is $\hat{\omega}\in \Omega$ such that $\xi(\hat{\omega})>0$. But then
\begin{align*}
\lim_{\lambda\rightarrow \infty} \Eb U(-\lambda \xi-\xi_1)
=-\lim_{\lambda\rightarrow \infty}\sum_{\omega\in \Omega}p(\omega)\ee^{\gamma (\lambda\xi(\omega)+\xi_1(\omega))}=-\infty<R,
\end{align*}
which implies the existence of $\lambda>0$ such that $\lambda \xi+\xi_1\notin B_{R}$, a contradiction.
\end{proof}

\begin{lemma}\label{thm:condstar} Fix $h\in \Rb^N$. Then, $\eta^*$ solves problem \eqref{eq:problem-3_entropic} if and only if 
\begin{align}
\Ec^X(h-\eta^*-c)-\Ebt \Ec^X(h-\eta^*-c)&=0,
	\label{eq:etastar1}\\
\Ec^Y(h-\eta^*-c)-\Ebt \Ec^Y(h-\eta^*-c)&=0
	\label{eq:etastar2}.
\end{align}
\end{lemma}
\begin{proof}
Note that, since problem \eqref{eq:problem-3} is a convex optimization problem when $U(x)=-\ee^{-\gamma x}$, the Karush-Kuhn-Tucker Theorem implies that conditions \eqref{eq:fstar1} and \eqref{eq:gstar1} are necessary and sufficient for the optimality of vectors $f^*$ and $g^*$. Suppose that $\eta^*$ is optimal for problem \eqref{eq:problem-3_entropic}. Let $f^*$ and $g^*$ be such that $\eta^*=f^*+g^*$ and $f^*\in \sigma(X_T)$, $g^*\in \sigma(Y_T)$. Then, $f^*$ and $g^*$ satisfy  \eqref{eq:fstar1} and \eqref{eq:gstar1}, which imply \eqref{eq:etastar1} and \eqref{eq:etastar2}. On the other hand, suppose that \eqref{eq:etastar1} and \eqref{eq:etastar2} hold. Again, let $f^*$ and $g^*$ be such that $\eta^*=f^*+g^*$ and $f^*\in \sigma(X_T)$, $g^*\in \sigma(Y_T)$. Then, $f^*$ and $g^*$ satisfy  \eqref{eq:fstar1} and \eqref{eq:gstar1}, which imply that $f^*$ and $g^*$ are optimal for problem \eqref{eq:problem-3}. Hence, $\eta^*$ is optimal for problem \eqref{eq:problem-3_entropic}.
\end{proof}
\begin{lemma}\label{thm:FixedPoint} Fix $h\in \Rb^N$. Then, $h^*\in B_{\Eb U(h-c)}$ is a fixed point of $\Hc$ if and only if 
\begin{align*}
\eta^* := h-h^*-c
\end{align*}
solves problem \eqref{eq:problem-3_entropic}.
\end{lemma}
\begin{proof}
Suppose that $h^*\in B_{\Eb U(h-c)}$ is a fixed point of $\Hc$, and set $\eta^*:=h-h^*-c$. Then,
\begin{align*}
\left[\Ec^X(h^*)-\Ebt \Ec^X(h^*)\right]
	+\left[\Ec^Y(h^*-\Ec^X(h^*)+\Ebt \Ec^X(h^*))
		-\Ebt \Ec^Y(h^*-\Ec^X(h^*)+\Ebt \Ec^X(h^*))\right]
		=0,
\end{align*}
which implies $\Ec^X(h^*)-\Ebt \Ec^X(h^*)\in \sigma (Y_T)$. Hence, 
\begin{align*}
h^* &= \Hc^Y(\Hc^X(h^*))\\
&=	h^*
	-\left[\Ec^X(h^*)-\Ebt \Ec^X(h^*))\right]
	-\left[\Ec^Y(h^*-\Ec^X(h^*)+\Ebt \Ec^X(h^*))
		-\Ebt\left(\Ec^Y(h^*-\Ec^X(h^*)
			+\Ebt \Ec^X(h^*))\right)\right]\\
&=	h^*-\Ec^Y(h^*)
		+\Ebt\left(\Ec^Y(h^*)\right).
\end{align*}
That is, \eqref{eq:etastar2} is satisfied. On the other hand, by Theorem \ref{thm:exp}, 
\begin{align*}
\Eb U(-h^*)
	& =\Eb U(-\Hc(h^*))\\
	& =\Eb U(-\Hc^Y(\Hc^X((h^*))))\\
	& \geq \Eb U(-\Hc^X(h^*))\geq \Eb U(-h^*),
\end{align*}
which implies $\Eb U(-\Hc^X(h^*))=\Eb U(-h^*)$. Then, if $\Hc^X(h^*)\neq h^*$, there are multiple solutions to the problem 
\begin{align}\label{eq:problem-3_f}
	&\sup_{f \in \Ac_X} \Eb U(f(X_T)
		+\eta^*(X_T,Y_T) - h(X_T,Y_T)+c) , &
\Ac_X
	&:=	\sigma(X_T)\cap
		\{ f : \Ebt f(X_T) = 0 \}.
\end{align}
However, this is a contradiction, since $\Ac_X$ is convex and the target function in problem \eqref{eq:problem-3_f} is strictly concave, and so \eqref{eq:problem-3_f} admits at most one solution. Hence, $\Hc^X(h^*)=h^*$, and \eqref{eq:etastar1} is also satisfied. Thus, Lemma \ref{thm:condstar} implies that $\eta^*$ is optimal. On the other hand, if $\eta^*$ is optimal, then Theorem \ref{thm:exp} implies $\Eb U(\Hc(h-\eta^*-c)) = \Eb U(h-\eta^*-c)$, and so if $h-\eta^*-c$ is not a fixed point of $\Hc$, problem \eqref{eq:problem-3_entropic} admits multiple solutions. This would, however, contradict the strict concavity of the map $\mathscr{U}$ defined by \eqref{eq:EU}.
\end{proof}

\noindent
The next theorem is the main result of this section.

\begin{theorem}\label{thm:expscheme} Problem \eqref{eq:problem-3_entropic} admits a unique solution $\eta^*$ given by $h-h^*-c$, where $h^*\in B_{\Eb U(h-c)}$ satisfies
\begin{align*}
h^*:=\lim_{n\rightarrow \infty} \Hc^n(h-c).
\end{align*}
\end{theorem}
\begin{proof}
Because, (1), the map $\mathscr{U}$ defined in \eqref{eq:EU} is strictly concave and continuous, (2), it satisfies 
\begin{align*}
\mathscr{U}(h-B_{\Eb U(h-c)}+c)\geq \mathscr{U}(0),
\end{align*}
and, (3), the set
\begin{align*}
h-B_{\Eb U(h-c)}-c=\mathscr{A}''\cap\{\eta:\Eb U(\eta-h+c)\geq \Eb U(h-c)\}
\end{align*}
is convex and compact, problem \eqref{eq:problem-3_entropic} admits a unique solution that belongs to $h-B_{\Eb U(h-c)}-c$. By Lemma \ref{thm:FixedPoint}, such solution must correspond to $h-h^*-c$, where $h^*$ is the unique fixed point of the operator $\Hc$ on $B_{\Eb U(h-c)}$. Consider the sequence $\{\Hc^n(h-c)\}_{n\in\Nb}\subset B_R$. Since $\{\Eb U(\Hc^n(h-c)\}_{n\in\Nb}$ is increasing (by Theorem \ref{thm:exp}) and bounded (since $\Eb U$ is continuous and $B_{\Eb U(h-c)}$ is compact), it admits a limit $e^*\in \Rb$. Since $B_{\Eb U(h-c)}$ is compact, any subsequence of $\{\Hc^n(h-c)\}_{n\in\Nb}$ admits a convergent subsequence. Let $\{\Hc^{n_{k}}(h-c)\}_{k\in \Nb}$ be one such convergent subsequence, and denote by $h^{**}$ its limit. By continuity of $\Hc$,
\begin{align*}
\lim_{k\rightarrow \infty} \Hc^{n_{k}+1}(h-c)
=\Hc(\lim_{k\rightarrow \infty} \Hc^{n_{k}}(h-c))
=\Hc(h^{**}).
\end{align*}
If $h^{**}$ is not optimal, Theorem \eqref{thm:exp} implies
\begin{align*}
\Eb U(\Hc(h^{**}))>\Eb U(h^{**}).
\end{align*}
On the other hand, $\Eb U(\Hc^{n_k+1}(h-c))$ and $\Eb U(\Hc^{n_k}(h-c))$ both converge to $e^*$, so it must be the case that $h^{**}$ is optimal, i.e., $h^{**}=h^*$. Hence, every subsequence of $\{\Hc^{n}(h-c)\}_{n\in \Nb}$ contains a further subsequence that converges to $h^*$. This implies that $\Hc^n(h-c)\rightarrow h^*$.
\end{proof}

\noindent
We have constructed a fixed point iteration for a Lipschitz map that is not a contraction nor nonexpansive. This result is similar to certain generalizations of the contraction mapping principle, which assume the existence of a partial order on the Banach space on which the iteration is performed (see \cite{Latif}).
\\[0.5em]
Based on Theorem \ref{thm:expscheme}, and assuming an equal split of the hedging capital among the two risks, one can set, for $N\in \Nb$,
\begin{align}
f^*_N:&=\frac{c}{2}+\sum_{n=0}^{N} \left[\Ec^X(\Hc^n(h-c))-\Ebt \Ec^X(\Hc^n(h-c))\right], \label{eq:fstarN}\\ 
g^*_N:&=\frac{c}{2}+\sum_{n=0}^{N} \left[\Ec^Y\left(\Hc^n(h-c)-\Ec^X(\Hc^n(h-c))-\Ebt \Ec^X(\Hc^n(h-c))\right)\right.\\
	& \quad \quad
	-\left.\Ebt\Ec^Y\left(\Hc^n(h-c)-\Ec^X(\Hc^n(h-c))-\Ebt \Ec^X(\Hc^n(h-c))\right)\right]. \label{eq:gstarN}
\end{align}
Although $f^*_N+g^*_N$ is guaranteed to converge to $\eta^*=h-h^*+c$, this does not imply that the sequences $\{f^*_N\}_{N\in\Nb}$ and $\{g^*_N\}_{N\in\Nb}$ will converge. This poses a computational challenge, since, if $f^*_N(\omega)\rightarrow +\infty$ for some $\omega\in \Omega$, and so also $g^*_N(\omega)\rightarrow -\infty$ because $f_N^*(\omega)+g^*_N(\omega)\rightarrow \eta^*(\omega)$, the numerical error in the implementation of \eqref{eq:fstarN} and \eqref{eq:gstarN} may become too large. The next result provides a condition for the boundedness of the sequences $\{f^*_N\}_{N\in\Nb}$ and $\{g^*_N\}_{N\in\Nb}$.

\begin{proposition}\label{thm:bounded} Suppose that 
\begin{align}\label{condition_prop}
\Pb(X=x,Y=y)>0 ,
\end{align}
for every $x\in \{X_T(\omega)\}_{\omega\in\Omega}$ and $y\in \{Y_T(\omega\}_{\omega\in\Omega}$. Then, the sequences $\{f_N^*\}_{N\in\Nb}$ and $\{g_N^*\}_{N\in\Nb}$, defined by expressions \eqref{eq:fstarN} and \eqref{eq:gstarN}, are bounded.
\end{proposition}
\begin{proof}
Note that, since $\|h-h^*\|_{\infty}<\infty$ and $\|h\|_{\infty}<\infty$, if 
\begin{align*}
f^*_{N_k}(\hat{\omega})\rightarrow
	 \pm \infty ,
\end{align*}
along any subsequence $\{N_k\}_{k\in\Nb}$ and for some $\hat{\omega}\in \Omega$, then it must be the case that 
\begin{align*}
g^*_{N_k}\rightarrow \mp \infty ,
\end{align*}
for every $\omega\in \Omega$, which in turn implies 
\begin{align*}
f^*_{N_k}(\hat{\omega})\rightarrow \pm \infty ,
\end{align*}
for every $\omega\in \Omega$. However, this is a contradiction, since 
\begin{align*}
\Ebt\left[ f_N^*\right]= \frac{c}{2} ,
\end{align*}
for every $N\in \Nb$. Similarly for $g_N^*$.
Hence, $\{f_N^*\}_{N\in\Nb}$ and $\{g_N^*\}_{N\in\Nb}$ are bounded, as requested.
\end{proof}

\noindent
Based on Proposition \ref{thm:bounded}, the solution to problem \eqref{eq:problem-3_entropic} can be constructed using equations \eqref{eq:fstarN} and \eqref{eq:gstarN}, and for $N\in \Nb$ such that the $\ell^2$ error
\begin{align*}
\|\Hc(f^*_N+g^*_N-h)-f^*_{N}+g^*_{N}-h\|_2 ,
\end{align*}
is lower than some threshold level $\epsilon$. We implement this approach in our numerical examples.
\\[0.5em]
The condition \eqref{condition_prop} in Proposition \ref{thm:bounded} is sufficient but not necessary. The result will hold as long as one can ``cycle'' through all the values $\{X_T(\omega),Y_T(\omega)\}_{\omega\in \Omega}$ as we will see in Example \ref{ExDiscrete1} and Remark \ref{cycle} below.

\begin{remark}\label{rmk:reflexive}
The extension of the results of this section to infinite dimensional spaces is a topological matter. We limit ourselves to observe that a natural requirement for the validity in infinite dimension of the results of this section, and of Theorem \ref{thm:expscheme} in particular, is that traded assets belong to a reflexive Banach space over which the entropic risk measure is lower semicontinuous with respect to the weak topology. For instance, $L^2$ is a good candidate as the entropic risk measure is lower semicontinuous with respect to the weak topology on $L^2$ (because it admits dual representation on $L^2$). Another good candidate is the Orlicz heart $M^{\Phi}$, defined by the Young function $\Phi(x)=\ee^{-\gamma x}-1$ (see \cite{CheriditoLi}). It is a reflexive Banach space and, again, the entropic risk measure admits dual representation on it.
\end{remark}

%\begin{remark}
%\red{A different approach is to approximate the $L^{\infty}$ problem on a discrete space, construct a solution and then take the limit as the approximation gets finer and finer.}
%\end{remark}

%\begin{remark}
%A fixed point iteration was constructed for a Lipschitz map that is not a contraction nor nonexpansive. This result is similar to certain generalizations of the contraction mapping principle, which assume the existence of a partial order on the Banach space on which the iteration is performed (see \cite{Latif}).
%\end{remark}

\subsubsection{Examples}

We begin with a simple example with only three states of the world to illustrate that one obtains the optimal solution by iterating the operator $\Hc$.

\begin{example}\label{ExDiscrete1}
Suppose $|\Omega|=3$, $\gamma=1$, $h(X_T,Y_T)= (X_TY_T-1)^+$, and
\begin{align*}
X_T &= 
\begin{bmatrix}
1 \\
2 \\
1
\end{bmatrix}, &
Y_T &= 
\begin{bmatrix}
1 \\
1 \\
4
\end{bmatrix}, &
h(X_T,Y_T) &= 
\begin{bmatrix}
0\\
1\\
3
\end{bmatrix}.
\end{align*}
Furthermore, suppose that $\Pbt=\Pb$, $\Pb(\omega_1) = 0.2, \
\Pb(\omega_2) = 0.3, \
\Pb(\omega_3) = 0.5$, so that
\begin{align*}
p_{X|Y} &= \begin{bmatrix}
.4 & .4 & 0\\
.6 & .6 & 0\\
0 & 0 & 1\\
\end{bmatrix}, &
p_{Y|X} &= \begin{bmatrix}
.2857 & 0 & .7143\\
0 & 1 & 0\\
.7143 & 0 & .2857\\
\end{bmatrix}.
\end{align*}
In this case, we expect a perfect solution to be found since $f(X_T)$ and $g(Y_T)$ are each defined up to 2 degrees of freedom, and $h$ is a three-dimensional vector (i.e., there are three equations to be satisfied in four unknowns). Iterating the operator $\Hc$ yields convergence in 16 steps with an $\ell^2$-error of less than $1\mathrm{e}{-10}$. See Figure \ref{ExpHedging}. The solution obtained is
\begin{align*}
f^*(X_T) &= \begin{bmatrix}
-0.3336\\
0.6664\\
-0.3336
\end{bmatrix}, &
g^*(Y_T) &= \begin{bmatrix}
0.3336\\
0.3336\\
3.3336
\end{bmatrix},
\end{align*}
and note that, in fact, a perfect hedge is achieved. 
%It is also worth noting that multiple solutions $f^*$ and $g^*$  in this case exist, since there are multiple solutions to the system of equations
%\begin{align*}
%\begin{bmatrix}
%0\\
%1\\
%3\\
%1.8
%\end{bmatrix}
%=\begin{bmatrix}
%1 & 0 & 1 & 0\\
%0 & 1 & 1 & 0\\
%1 & 0 & 0 & 1\\
%0.7 & 0.3 & 0.5 & 0.5 
%\end{bmatrix}
%\begin{bmatrix}
%f(X_T(\omega_1))\\
%f(X_T(\omega_2))\\
%g(X_T(\omega_1))\\
%g(X_T(\omega_3))
%\end{bmatrix}.
%\end{align*}
\\[0.5em]

\begin{remark}\label{cycle}
It is worth noting that, in the case considered in Example \ref{ExDiscrete1}, the condition in Proposition \ref{thm:bounded} does not hold. However, it is easy to see that the result of Proposition \ref{thm:bounded} is still true, since if, say, $g^*_N(\omega_3)\rightarrow \infty$, then $f^*_N(\omega_3)\rightarrow -\infty$, which implies $f^*_N(\omega_1)\rightarrow -\infty$, and so $g^*_N\rightarrow \infty$, a contradiction. Hence, we can still utilize the procedure outlined above and based on Theorem \ref{thm:expscheme}.
\end{remark}

\begin{remark}
It is also worth pointing out that the solution is unique once the capital allocation is fixed, since the system
\begin{align*}
\begin{bmatrix}
0\\
1\\
3\\
-0.0336\\
1.8336
\end{bmatrix}
=
\begin{bmatrix}
1 & 0 & 1 & 0\\
0 & 1 & 1 & 0\\
1 & 0 & 0 & 1\\
0.7 & 0.3 & 0 & 0\\
0 & 0 & 0.5 & 0.5 
\end{bmatrix}
\begin{bmatrix}
f(X_T(\omega_1))\\
f(X_T(\omega_2))\\
g(X_T(\omega_1))\\
g(X_T(\omega_3))
\end{bmatrix}
\end{align*}
admits a unique solution.
\end{remark}
\begin{figure}
\begin{center}
\begin{tabular}{cc}
\includegraphics[width=0.48\textwidth]{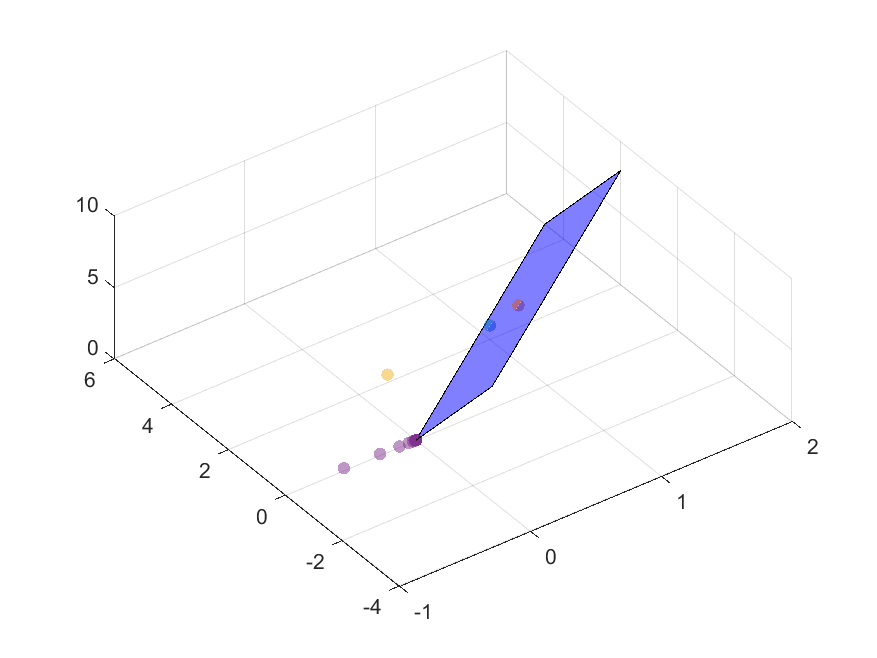}
\end{tabular}
\caption{
The yellow dot is the vector $h$. The blue plane is the space spanned by $X_T$ and $Y_T$, which are represented by the blue and red dots. The purple dots represent the sequence of portfolios obtained by iterating the operator $\Hc$. When the algorithm stops, the coordinates of the last purple dot are in the proximity of the origin $(0,0,0)$, showing that a perfect hedge is achieved.}\label{ExpHedging}
\end{center}
\end{figure}
\end{example}

\noindent
In our next example we consider a more realistic situation in which logarithmic returns on the underlying assets are assumed to be normally distributed. An approximation through simple random variables is in this case required to implement the iteration scheme above derived.

\begin{example} Consider now a claim paying off $h(X_T,Y_T):=(\ee^{X_T+Y_T}-1)^+$, with $(X_T,Y_T)\sim \mathscr{N}(\mu,\Sigma)$, where 
\begin{align}\label{eq:muSigma}
\mu &= \begin{bmatrix}
0.1\\
0.15
\end{bmatrix}, & 
\Sigma &= \begin{bmatrix}
1 & \sqrt{2}\rho \\
\sqrt{2}\rho & 2
\end{bmatrix},
\end{align}
and $\rho \in [-1,1]$ satisfies
\begin{align*}
\rho = \frac{Cov(X_T,Y_T)}{\sqrt{2}}.
\end{align*}
Assume again that $\Pbt = \Pb$ for simplicity (if $\Pb\neq \Pbt$, one simply needs to add the Radon Nikodym derivatives $p_X/\tilde{p}_X$ and $p_Y/\tilde{p}_Y$ to the solutions $f$ and $g$). Define
\begin{align*}
x^N &= \{-N+\delta n\}_{n=0,...,2N/\delta}, &
y^N &= \{-N+\delta n\}_{n=0,...,2N/\delta},
\end{align*}
where $N\in \Nb$ and $\delta>0$ are fixed, and set, for every $i,j\in \{0,...,2N/\delta\}$,
\begin{align*}
p_{X,Y}^N(i,j)
&:=\Pb\left(X_T\in [x^N_i,x^N_{i+1}), Y_T\in [y^N_j,y^N_{j+1})\right), \\ 
p_X^N(i) &:= \Pb\left(X_T\in [x^N_i,x^N_{i+1})\right), \
p_Y^N(j)  := \Pb\left(Y_T\in [y^N_i,y^N_{j+1})\right).
\end{align*}
Observe that, in this case, the condition in Proposition \ref{thm:bounded} holds as long as the correlation $\rho$ between $X_T$ and $Y_T$ satisfies $|\rho|<1$.\\[0.5em]
Proceeding analogously as in Example \ref{ExDiscrete1}, one can define
\begin{align*}
p_{X|Y}^N(i,j)
&:=\frac{\Pb\left(X_T\in [x^N_i,x^N_{i+1}),
				Y_T\in [y^N_j,y^N_{j+1})\right)}
	{\Pb\left(Y_T\in [y^N_j,y^N_{j+1})\right)}
 =\frac{p_{X,Y}^N(i,j)}{p^N(j)},\\
p_{Y|X}^N(j,i)
&:=\frac{\Pb\left(X_T\in [x^N_i,x^N_{i+1}),
				Y_T\in [y^N_j,y^N_{j+1})\right)}
	{\Pb\left(X_T\in [x^N_j,x^N_{j+1})\right)}
 =\frac{p_{X,Y}^N(i,j)}{p^N(i)}.
\end{align*}
from which the operator $\Hc$ can be constructed. 
\\[0.5em]
The iteration scheme defined by Theorem \ref{thm:expscheme} and equations \eqref{eq:fstarN} and \eqref{eq:gstarN} was here implemented assuming $\rho=-0.1$, $c=0$ and $N=4$.\footnote{Note that $\Pb(Y_T<-4)\approx 0.02$ and $\Pb(Y_T>4)\approx 0.03$, and similarly for $X_T$. While this is an illustrative example, higher precision may be obtained by truncating the random vector $(X_T,Y_T)$ along a level curve of its cumulative distribution function.} Furthermore, we considered $\gamma \in \{0.05, 0.1, 1, 10, 20\}$. The optimal hedges $f^*$ and $g^*$ are shown in Figure \ref{fig:ExpHedgingNormal}(a)-(b). Note that lower values of $\gamma$, i.e. lower risk aversion, provide positive profits in states of higher probability at the cost of higher losses in states with lower probability. 
\\[0.5em]
Figure \ref{fig:ExpHedgingNormal}(c)-(d) depicts the functions $f^*$ and $g^*$ for different values of the correlation coefficient $\rho$ between $X_T$ and $Y_T$. As mentioned, when $\rho \in\{-1,1\}$, the condition and the result of Proposition \ref{thm:bounded} do not hold. In fact, for values of $\rho$ that are close to $\pm 1$, our algorithm fails to converge due to the fact that $p^N_{X,Y}(i,j)\approx 0$ for some values of $i,j\in \{0,...,2N/\delta\}$. 

\begin{figure}
\centering
\begin{subfigure}{0.4\textwidth}
    \includegraphics[width=\textwidth]{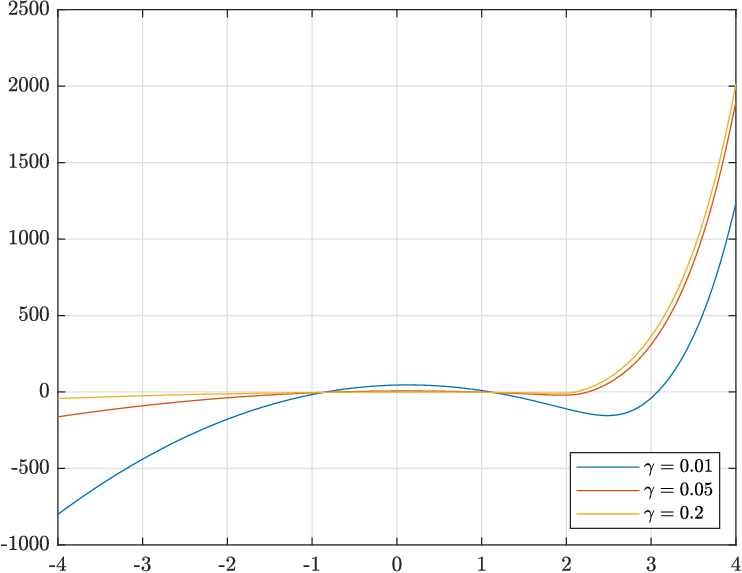}
    \caption{}
\end{subfigure}
\hfill
\begin{subfigure}{0.4\textwidth}
    \includegraphics[width=\textwidth]{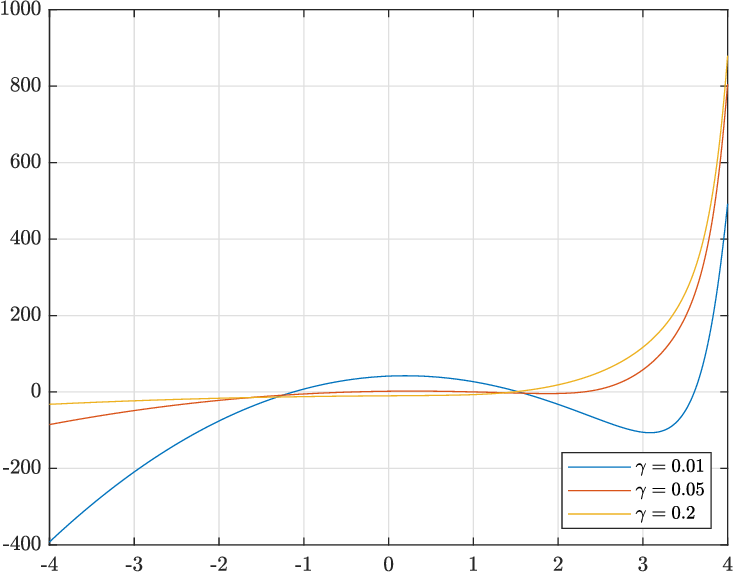}
    \caption{}
\end{subfigure}
\centering
\begin{subfigure}{0.4\textwidth}
    \includegraphics[width=\textwidth]{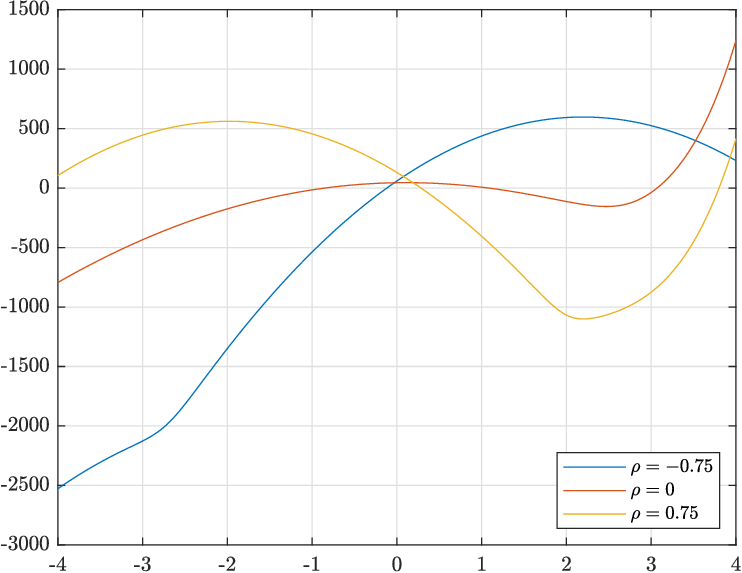}
    \caption{}
\end{subfigure}
\hfill
\begin{subfigure}{0.4\textwidth}
    \includegraphics[width=\textwidth]{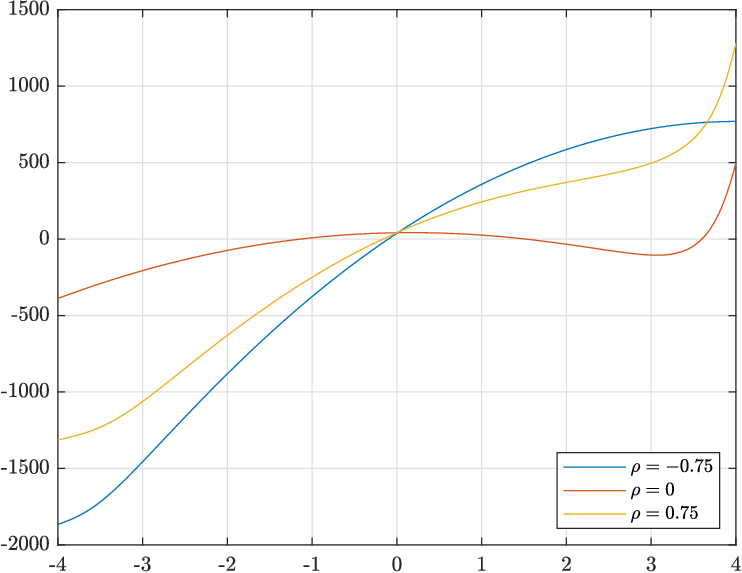}
    \caption{}
\end{subfigure}
\caption{Figure (a) and Figure (b) show, respectively, the optimal $f^*$  and $g^*$ for problem \eqref{eq:problem-3} with exponential utility and different values of $\gamma$, and assuming that $X_T$ and $Y_T$ are jointly normally distributed with $\mu$ and $\Sigma$ as in equation \eqref{eq:muSigma} and $\rho=-0.1$. Figure (c) and Figure (d) depict the optimal $f^*$ and $g^*$ for different values of the correlation $\rho$ between $X_T$ and $Y_T$, assuming that they are jointly normally distributed with $\mu$ and $\Sigma$ as in equation \eqref{eq:muSigma} and that utility is exponential with $\gamma = 0.01$.}
\label{fig:ExpHedgingNormal}
\end{figure}

\end{example}

\subsection{Mean-variance static hedging}

The first order conditions are also tractable in the case of quadratic utility, as shown in the following result.

\begin{theorem}[Mean-variance hedging]
Suppose the function $U$ in problem \eqref{eq:problem-3} is a concave quadratic function $U(x) = \gam x -  \tfrac{1}{2} x^2$.
Then, the pair of optimized hedges, $f^*$ and $g^*$, satisfies
\begin{align}
%0
	%&=	\frac{ \del L }{\del f}(\lam^*,f^*,g^*) &
	%&\Rightarrow &
f^*(x)
	%&=	\frac{1}{p_X(x)} \int \dd y \, p_{X,Y}(x,y) \big( h(x,y) - g^*(y) \big) + \lam^* \frac{ \pt_X(x) }{p_X(x)} + a \\
	&=	\Eb[ h(X_T,Y_T) - g^*(Y_T) | X_T = x ] + \lam^* \frac{ \pt_X(x) }{p_X(x)} + \gam , \label{eq:fstar-quad} \\
%0
	%&=	\frac{ \del L }{\del g}(\lam^*,f^*,g^*) &
	%&\Rightarrow &
g^*(y)
	%&=	\frac{1}{p_Y(y)} \int \dd x \, p_{X,Y}(x,y) \big( h(x,y) - f^*(x) \big) + \lam^* \frac{ \pt_Y(y) }{p_Y(y)} + a \\ 
	&=	\Eb[ h(X_T,Y_T) - f^*(X_T) | Y_T = y ] + \lam^* \frac{ \pt_Y(y) }{p_Y(y)} + \gam , \label{eq:gstar-quad}
\end{align}
where $\lam^*$ is given by either
\begin{align}
\lam^*
	&=	0 , &
	&\text{or}&
\lam^*
	&=	\frac{c - 2\gam - \Ebt \Eb[ h(X_T,Y_T) - g^*(Y_T) | X_T ] -  \Ebt \Eb[ h(X_T,Y_T) - f^*(X_T) | Y_T ] }
			{\Ebt(\pt_X(X_T)/p_X(X_T)) + \Ebt( \pt_Y(Y_T)/p_Y(Y_T)) \Big) } . \label{eq:options2}
\end{align}
\end{theorem}

\begin{proof}
Comparing \eqref{eq:optimality-1} and \eqref{eq:opt-f}, we see that $h(x,y)$ in the former has been replaced by $h(x,y) - g^*(y)$ in the latter.  As such, we find from \eqref{eq:f-star-quad} that $f^*$ and $g^*$ satisfy \eqref{eq:fstar-quad} and \eqref{eq:gstar-quad}, respectively, where the expression for $g^*$ follows by comparing \eqref{eq:opt-f} and \eqref{eq:opt-g}.
%\begin{align}
%%0
	%%&=	\frac{ \del L }{\del f}(\lam^*,f^*,g^*) &
	%%&\Rightarrow &
%f^*(x)
	%&=	\frac{1}{p_X(x)} \int \dd y \, p_{X,Y}(x,y) \big( h(x,y) - g^*(y) \big) + \lam^* \frac{ \pt_X(x) }{p_X(x)} + \gam \\
	%&=	\Eb[ h(X_T,Y_T) - g^*(Y_T) | X_T = x ] + \lam^* \frac{ \pt_X(x) }{p_X(x)} + \gam , \label{eq:fstar-quad} \\
%%0
	%%&=	\frac{ \del L }{\del g}(\lam^*,f^*,g^*) &
	%%&\Rightarrow &
%g^*(y)
	%&=	\frac{1}{p_Y(y)} \int \dd x \, p_{X,Y}(x,y) \big( h(x,y) - f^*(x) \big) + \lam^* \frac{ \pt_Y(y) }{p_Y(y)} + \gam \\ 
	%&=	\Eb[ h(X_T,Y_T) - f^*(X_T) | Y_T = y ] + \lam^* \frac{ \pt_Y(y) }{p_Y(y)} + \gam , \label{eq:gstar-quad}
%\end{align}
Next, from \eqref{eq:lambda-condition-2}, we must have either
\begin{align}
0
	&=	\lam^* , &
	&\text{or}&
c
	&=	\Ebt f^*(X_T) + \Ebt g^*(Y_T) . \label{eq:options}
\end{align}
inserting \eqref{eq:fstar-quad} and \eqref{eq:gstar-quad} into the second equation in \eqref{eq:options} we have
\begin{align}
c
	&=	\Ebt \Eb[ h(X_T,Y_T) - g^*(Y_T) | X_T ] + \lam^* \Ebt \Big( \frac{ \pt_X(X_T) }{p_X(X_T)} \Big)  + \gam \\ &\quad
			+ \Ebt \Eb[ h(X_T,Y_T) - f^*(X_T) | Y_T ] + \lam^* \Ebt \Big( \frac{ \pt_Y(Y_T) }{p_Y(Y_T)} \Big) + \gam .
\end{align}
Solving for $\lam^*$ yields the second equality in \eqref{eq:options2}.
\end{proof}

%\begin{example}
%\blu{insert example here}.
%\end{example}

%\blu{*********** Begin Notes for Yoshi *********************}
%\\[0.5em]
\noindent
The expressions for the optimizers in \eqref{eq:fstar-quad} and \eqref{eq:gstar-quad} provide the basis for a numerical scheme to compute the solutions. Also, we note that, by inserting \eqref{eq:gstar-quad} into \eqref{eq:fstar-quad} and viceversa, we obtain
\begin{align}
f^*(x)
	&=	\Eb \Big[ h(X_T,Y_T) \Big| X_T = x \Big] + \lam^* \frac{ \pt_X(x) }{p_X(x)} 
			-  \lam^* \Eb \Big[ \frac{ \pt_Y(Y_T) }{p_Y(Y_T)} \Big| X_T = x \Big] \\ &\quad
			- \Eb \Big[ \Eb \big[ h(X_T,Y_T) - f^*(X_T) \big| Y_T \big] \Big| X_T = x \Big] , \label{eq:fstar-prob} \\
g^*(y)
	&=	\Eb \Big[ h(X_T,Y_T) \Big| Y_T = y \Big] + \lam^* \frac{ \pt_Y(y) }{p_Y(y)}
			- \lam^* \Eb \Big[ \frac{ \pt_X(X_T) }{p_X(X_T)} \Big| Y_T = y \Big] \\ &\quad
			- \Eb \Big[ \Eb \big[ h(X_T,Y_T) - g^*(Y_T) \big| X_T \big] \Big| Y_T = y \Big], \label{eq:gstar-prob}
\end{align}
so that, at least when $\lam^*=0$, $f^*$ and $g^*$ can be defined as fixed points of a given map.

\section{Indifference pricing}\label{sec:indiff}

So far, in both of the static hedging problems considered, the price $c$ of the claim to be hedged is given or known. This assumption is reasonable, as one can think of a situation with multiple trading desks specialized in different markets and each providing quotes for the instruments traded. Such quotes can be generated quickly though existing methodologies (in both static and dynamic settings), such as conic finance (see \cite{MadanCherny}, \cite{Shirai}) or good deal bounds (see \cite{Cochrane}). These methods can also be used for the purpose of hedging in incomplete markets, but that requires solving a challenging optimization problem. Our results are then useful as they can be applied to the portfolio composed of each desk's position and are simpler to implement.\\[0.5em]
A natural question remains, however, for the trading parties, i.e., what is the price implied by our hedging methodology. To that end, we apply the notion of indifference pricing (see \cite{Carmona} and references therein) in the context of static hedging. The indifference price can be interpreted as the investor's reservation value for the claim $h$, given that it will be statically and partially hedged using $f$. Let us define the \textit{indifference price} $p$ for $\nu$ derivatives, each paying $h(X_T,Y_T)$, as the unique solution of the following equation
\begin{align}
\sup_{f \in \Ac(c-p  \nu)} \Eb \, U(f(X_T) + \nu  h(X_T,Y_T))
	&=	\sup_{f \in \Ac(c)} \Eb \, U(f(X_T)) , &
\Ac(c)
	&=	\{ f : \Ebt f(X_T) = c \} . \label{eq:indiff-step0}
\end{align}
Here, $\nu$ may be positive (indicating a buyer's price) and negative (indicating a seller's price). Hence, the indifference pricing approach can generate bid and ask prices, which can be compared to observed market prices, when available.
\\[0.5em]
With an example, we now illustrate how an analytical solution for the indifference price can be obtained in our setup.
\begin{example}
Suppose the utility function is exponential $U(x) = -\ee^{- \gam x}/\gam$.  Then taking $h(x,y) = 0$ we have from \eqref{eq:f-star-exp} that
\begin{align}
\sup_{f \in \Ac(c)} \Eb \, U(f(X_T))
	&=	\Eb \, U(f_1^*(X_T)) , &
f_1^*(x)
	&= c + \frac{1}{\gam} \log \Big( \frac{p_X(x)}{\pt_X(x)} \Big) - \frac{1}{\gam} \int \dd \xi \, \pt_X(\xi) \log \Big( \frac{p_X(\xi)}{\pt_X(\xi)} \Big) , 
\end{align}
and thus
\begin{align}
\sup_{f \in \Ac(c)} \Eb \, U(f(X_T))
	%&=	\frac{- \ee^{- \gam c}}{\gam} \Eb \Big( \frac{\pt_X(X_T)}{p_X(X_T)} \Big) 
			%\exp \int \dd z \, \pt_X(z) \log \Big( \frac{p_X(z)}{\pt_X(z)} \Big) \\ 
	&=	\frac{- \ee^{- \gam c}}{\gam} \exp \Big( \int \dd \xi \, \pt_X(\xi) \log \Big( \frac{p_X(\xi)}{\pt_X(\xi)} \Big) \Big) . \label{eq:indiff-step1}
\end{align}
Similarly, we have from \eqref{eq:f-star-exp} with $c \to c - p  \nu$ and $h(x,y) \to \nu  h(x,y)$ that
\begin{align}
\sup_{f \in \Ac(c-p  \nu)} \Eb \, U(f(X_T) + \nu  h(X_T,Y_T))
	&= \Eb \, U(f_2^*(X_T) + \nu  h(X_T,Y_T)) , \\
f_2^*(x)
	&=	c - p  \nu + \frac{1}{\gam} \log \Big( \frac{1}{\pt_X(x)}  \int \dd y \, p_{X,Y}(x,y) \ee^{\gam \nu  h(x,y)} \Big) \\ &\quad
			- \frac{1}{\gam} \int \dd \xi \, \pt_X(\xi) \log \Big( \frac{1}{\pt_X(\xi)} \int \dd y \, p_{X,Y}(\xi,y) \ee^{\gam \nu  h(\xi,y)}  \Big) .
\end{align}
and thus
\begin{align}
\sup_{f \in \Ac(c-p  \nu)} \Eb \, U(f(X_T) + \nu  h(X_T,Y_T)) 
	&=	\frac{- \ee^{- \gam (c - p  \nu)}}{\gam}
			\exp \Big( \int \dd \xi \, \pt_X(\xi) \log \Big( \frac{1}{\pt_X(\xi)} \int \dd y \, p_{X,Y}(\xi,y) \ee^{\gam \nu  h(\xi,y)}  \Big) \Big) \\ &\quad
			\int \dd \xi \, p_X(\xi) \Big( \frac{1}{\pt_X(\xi)}  \int \dd y \, p_{X,Y}(\xi,y) \ee^{\gam \nu  h(\xi,y)} \Big)^{-1}.
	\label{eq:indiff-step2}
\end{align}
Inserting \eqref{eq:indiff-step1} and \eqref{eq:indiff-step2} into \eqref{eq:indiff-step0} and solving for $\gam p \nu$ we obtain
\begin{align}\label{eq:indiffp_formula}
\gam p \nu
	&=	\int \dd \xi \, \pt_X(\xi) \log \Big( \frac{p_X(\xi)}{\pt_X(\xi)} \Big)
			- \int \dd \xi \, \pt_X(\xi) \log \Big( \frac{1}{\pt_X(\xi)} \int \dd y \, p_{X,Y}(\xi,y) \ee^{\gam \nu  h(\xi,y)}  \Big) \\ &\quad
			- \log \int \dd \xi \, p_X(\xi) \Big( \frac{1}{\pt_X(\xi)}  \int \dd y \, p_{X,Y}(\xi,y) \ee^{\gam \nu  h(\xi,y)} \Big)^{-1} .
\end{align}
We note that $p$ in \eqref{eq:indiffp_formula} does not depend on $c$ since the utility function $U$ is in exponential form. More generally, the indifference price under other utility functions may depend on $c$.
%\\ 
%\blu{(Here, it would be interesting to add a plot of $p$ as a function of $\gam \nu$ -- perhaps in the same setting as Example \ref{ex:LETF}.)}
\end{example}

\noindent
In addition to the indifference price, we also obtain the optimal static hedging $f^*$. In the literature, indifference pricing has been applied to determine the quantity of options in the static hedge; see \cite{IlhanSircar,LeungSircar}. In contrast, our approach derives the optimal general form of the static hedging position, instead of optimizing over a single parameter (i.e., quantity).

\section{Conclusion}
\label{sec:conclude}
In this paper, a static hedge for a portfolio of derivatives that depend on multiple underlying assets is constructed as the solution to an infinite dimensional optimization problem in which the objective function is the expected utility of the final position. The first order conditions are relatively easily obtained via variational methods. The resulting equations can be solved analytically when a single underlying risk is being hedged under exponential, logarithmic, power, or quadratic utility. The optimal hedge can then be implemented in practice by trading European put and call options written on the underlying asset.
\\[0.5em] 
In the case with two underlying assets, an iterative scheme for a (possibly expansive) Lipschitz map is shown to converge to the optimal position, assuming exponential utility and finite state space. Finally, we show how the scheme can be practically used by truncating and approximating with simple random variables a two dimensional normally distributed random vector.
\\[0.5em]
This paper also gives rise to some new questions for future research. A natural next step is to expand the risk exposure to include additional underlying assets or sources of risks, or to design static hedging portfolios for specific exotic options, such as Asian and other path-dependent options. The static hedging framework can be extended to a semi-static setup whereby static positions are updated at one or several random times in the future. Finally, it would be of practical interest to further examine the interplay between a static hedge and indifference price in the context of pricing and trading a contingent claim in incomplete markets.

\end{document}